\documentclass[11pt]{article}

\DeclareUnicodeCharacter{0301}{\'e}
\usepackage[margin=1in]{geometry}
\usepackage{multirow}
\usepackage{amssymb,amsmath,amsthm,amsfonts}
\usepackage{bm}
\usepackage{textgreek}
\usepackage{mathtools}
\usepackage{enumitem}
\usepackage[numbers,comma,sort&compress]{natbib}
\usepackage{authblk}
\usepackage{graphicx}
\usepackage[font=small]{caption}
\usepackage{subcaption} 
\usepackage{tablefootnote} 
\usepackage{float}
\usepackage{placeins}
\usepackage[ruled,vlined,linesnumbered]{algorithm2e}
\usepackage{algorithmic}

\usepackage{physics}
\usepackage{footnote}
\usepackage{xcolor}
\usepackage{mathrsfs}
\usepackage{bbm}
\usepackage{makecell}
\usepackage{pbox}
\usepackage[colorlinks]{hyperref}
\hypersetup{citecolor=blue}
\newtheorem{theorem}{Theorem}
\newtheorem{definition}{Definition}
\newtheorem{lemma}{Lemma}

\newtheorem{assumption}{Assumption}

\newcommand{\eq}[1]{(\ref{eq:#1})}

\newcommand{\thm}[1]{\hyperref[thm:#1]{Theorem~\ref*{thm:#1}}}
\newcommand{\cor}[1]{\hyperref[cor:#1]{Corollary~\ref*{cor:#1}}}
\newcommand{\defn}[1]{\hyperref[defn:#1]{Definition~\ref*{defn:#1}}}
\newcommand{\lem}[1]{\texorpdfstring{\hyperref[lem:#1]{Lemma~\ref*{lem:#1}}}{Lemma~\ref*{lem:#1}}}
\newcommand{\prop}[1]{\hyperref[prop:#1]{Proposition~\ref*{prop:#1}}}
\newcommand{\assum}[1]{\hyperref[assum:#1]{Assumption~\ref*{assum:#1}}}
\newcommand{\fig}[1]{\hyperref[fig:#1]{Figure~\ref*{fig:#1}}}
\newcommand{\tab}[1]{\hyperref[tab:#1]{Table~\ref*{tab:#1}}}
\newcommand{\algo}[1]{\hyperref[algo:#1]{Algorithm~\ref*{algo:#1}}}
\renewcommand{\sec}[1]{\hyperref[sec:#1]{Section~\ref*{sec:#1}}}
\newcommand{\append}[1]{\hyperref[append:#1]{Appendix~\ref*{append:#1}}}
\newcommand{\fac}[1]{\hyperref[fac:#1]{Fact~\ref*{fac:#1}}}
\newcommand{\lin}[1]{\hyperref[lin:#1]{Line~\ref*{lin:#1}}}

\renewcommand{\arraystretch}{1.25}

\def\>{\rangle}
\def\<{\langle}

\newcommand{\ex}{\mathrm{ex}}

\usepackage{cleveref}
\usepackage[utf8]{inputenc}
\usepackage{authblk}

\begin{document}

\title{Time-Dependent Low-Energy Simulation \\ Accelerates Adiabatic State Preparation}
\author[1]{Shuo Zhou}
\author[2]{Zhaokai Pan}

\author[3]{Weiyuan Gong}

\author[1]{Tongyang Li\thanks{Corresponding author. Email: \href{mailto:tongyangli@pku.edu.cn}{tongyangli@pku.edu.cn}}}

\affil[1]{CFCS, School of Computer Science, Peking University}
\affil[2]{IIIS, Tsinghua University}
\affil[3]{SEAS, Harvard University}


\date{}
\maketitle

\begin{abstract}
Hamiltonian simulations are key subroutines in adiabatic quantum computation and quantum many-body physics, where quantum dynamics often happen in the low-energy sector. Previous studies have shown that the low-energy assumption can reduce the resource requirements of standard time-independent Hamiltonian simulation algorithms. However, whether such advantages extend to \emph{time-dependent} Hamiltonian simulation remains open.
In this paper, we consider the adiabatic regime where the relevant low-energy subspace is spanned by a fixed number of low-energy eigenstates and separated from the rest of the spectrum by a gap.
We show that, for simulating spin Hamiltonians by product formulas, the explicit system size dependence in the leading commutator-scaling term can be replaced by a low-energy scale up to logarithmic factors. 
Technically, we derive the low-energy simulation error with commutator scaling for product formulas by leveraging adiabatic perturbation theory to analyze the time-variant energy spectrum of the underlying Hamiltonian.
We further conduct numerical experiments on adiabatic state preparation of an illustrative example system to support our theoretical findings.
Finally, we prove a lower bound of query complexity for generic time-dependent Hamiltonian simulations.
\end{abstract}

\section{Introduction}
Simulating the dynamics of quantum systems governed by an underlying Hamiltonian is one of the primary applications of quantum computers~\cite{feynman1982simulating}. 
Since the quantum simulation algorithms were first studied by Lloyd~\cite{doi:10.1126/science.273.5278.1073} using product formulas~\cite{suzuki1985decomposition,suzuki1990fractal,suzuki1991general}, numerous advanced techniques~\cite{berry2007efficient, childs2012hamiltonian, berry2015simulating, berry2014exponential, Berry_2015, Low_2017, low2019hamiltonian} have been developed for time-independent Hamiltonian simulations. 
However, simulations of time-dependent Hamiltonians have been significantly less studied, with techniques ranging from time-dependent product formulas~\cite{JHuyghebaert_1990,Wiebe_2010}, Dyson series~\cite{low2019hamiltoniansimulationinteractionpicture,Berry_2020} to Magnus expansion~\cite{fang2025highordermagnusexpansionhamiltonian}, etc. 
Given an initial state $\ket{\psi(0)}$, the time evolution is governed by the Schr\"{o}dinger's equation: 
\begin{align}\label{eq:Schr}
i\partial_t \ket{\psi(t)} = H(t) \ket{\psi(t)},~t\in [0,T].
\end{align}
Our goal is to compute the quantum state $\ket{\psi(T)}$ at the final stage. 
To achieve this, we construct a quantum circuit to approximate the exact evolution operator with the time-ordering operator $\mathcal{T}$:
\begin{align}
V(T,0)=\mathcal{T}\exp\left(-i\int_0^TH(\tau)\dd{\tau}\right).
\end{align}

The applications of time-dependent Hamiltonian simulations range from adiabatic quantum computation~\cite{farhi2000quantumcomputationadiabaticevolution} and quantum control~\cite{Dong_2010}, to quantum many-body physics~\cite{Aspuru_Guzik_2005}, where quantum dynamics often take place in the low-energy sector. 
In the simpler case where the Hamiltonian is time-independent, this problem has attracted much attention from the community. 
Former studies~\cite{csahinouglu2021hamiltonian, Gong_2024, hejazi2024better} have shown that when the initial state is only supported on a low-energy subspace, the simulation complexity of product formulas~\cite{doi:10.1126/science.273.5278.1073,campbell2019random,childs2019faster,tran2021faster} can be significantly improved. 
Ref.~\cite{Zlokapa_2024} further proposed low-energy simulation algorithms based on quantum singular value transformation~\cite{Gily_n_2019}. 
Recently, Ref.~\cite{Mizuta_2025} substantially improved the low-energy simulation analysis for product formulas with commutator scaling~\cite{childs2021theory}. 
However, the understanding of quantum simulation algorithms for time-dependent Hamiltonians under the low-energy assumptions is limited.

Among these approaches, we focus on time-dependent product formulas. Given a Hamiltonian in the form of a summation $H(t)=\sum_{\gamma=1}^\Gamma H_\gamma(t)$, the first-order Trotterization decomposes $T$ into equidistant steps $\delta=T/r$, and then implements the product of (time-ordered) exponentials of $H_\gamma(t)$.
For large enough $r$, we have
\begin{align*}
V(T,0)\approx \Pi_{j=1}^r \Pi_{\gamma=1}^{\Gamma}\mathcal{T}e^{-i\int_{(j-1)\delta}^{j\delta} H_\gamma(\tau)\dd{\tau}}.
\end{align*}
We can further generalize the product formula to $(p\geq 2)$-th orders using the recursion rule~\cite{suzuki1990fractal, suzuki1991general}.
There are two different common choices of product formula with similar properties in each time step in recent literature~\cite{JHuyghebaert_1990, Wiebe_2010}, known as generalized and standard time-dependent product formulas.
The main challenge for implementing product formulas is to decide a proper $r$ to guarantee that the simulation error is below the target threshold $\epsilon$.

This paper addresses this gap for smooth adiabatic evolutions of low-energy states. The analysis combines commutator-scaling product-formula bounds with an adiabatic leakage estimate, allowing us to follow the changing instantaneous low-energy subspace throughout the evolution.

Unlike the time-independent case, the exact evolution operator with time-ordered exponential no longer commutes with the low-energy subspace projector. 
We cannot adopt the analysis framework in the time-independent case, splitting the error into projected error and effective low-energy norm. 
To address this issue, we define the relevant low-energy sector by a fixed number of instantaneous eigenstates and use adiabatic perturbation theory as a key subroutine to compute the error terms in the analysis framework.
The resulting bounds apply naturally to adiabatic state preparation and give a low-energy product-formula guarantee for a broad class of smooth time-dependent spin Hamiltonians.

\subsection{Main results}
We now state the main Trotter-number guarantee in the adiabatic setting, given the gap and smoothness assumptions that are natural in various quantum systems.
Our task is to simulate an $N$-qubit, $k$-local multi-linear Hamiltonian, which has a fixed structure and time-variant coefficients (see Eq.~\eqref{eq:multilinear} for the formal definition), by product formulas when the initial state is only supported on a low-energy subspace spanned by $\sigma$ eigenstates. 
We assume that this low-energy subspace has an energy spectrum within the regime $[0,\Delta]$ over time $[0,T]$, and is separated from the remaining spectrum by a gap $\gamma$. 
We also assume that the Hamiltonian has bounded derivatives $\dot{H}\coloneqq \tfrac{\dd{H}}{\dd{s}}$ and $\ddot{H}\coloneqq \tfrac{\dd^2{H}}{\dd{s^2}}$ with respect to the scaled time $s=t/T$, and its interaction strength on every single qubit is upper bounded by $g$ in \eqref{eq:g}. 
We present our main results in the following \thm{main}, whose formal description is given by \thm{atrotternumber} in \sec{application}. We summarize our results in \tab{result_table}.

\begin{theorem}[Informal]\label{thm:main} Consider a smooth adiabatic interpolation whose evolving low-energy subspace remains separated from the rest of the spectrum by a gap $\gamma$. For all evolution times $T = \Omega(\frac{gN}{\Delta}\cdot(\sigma\frac{\|{\dot{H}}\|}{\gamma^2}+\sigma\sqrt{\sigma} \frac{\|{\dot{H}}\|^2}{\gamma^3}+\sigma\frac{\|{\ddot{H}}\|}{\gamma^2}))$, the Trotter number of both $p$-th order generalized and standard time-dependent product formulas is
\begin{equation*}
     r = \tilde{O}\left(g\frac{(\Delta+g\log(N/\epsilon))^{1/p}T^{1+1/p}}{\epsilon^{1/p}}\right),
\end{equation*}
for initial states supported in the low-energy subspace. 
\end{theorem}

\noindent Compared with the full Hilbert space, the dependence on the system size $N$ is partially substituted by the low-energy threshold supremum $\Delta$, as seen in the first term of $r$. 
We obtain the maximum improvement compared to simulations in the full Hilbert space when the low-energy scale $\Delta+g\log(N/\epsilon)$ is much smaller than the full extensive scale $gN$.
We remark that when $\dot{H}=0$, the requirement on $T$ reduces to a tautology, reproducing the time-independent results~\cite{Mizuta_2025}, which is self-contained with the fact that time-independent Hamiltonians are a special case of time-dependent Hamiltonians. 

We also prove a lower bound compatible with this fact that (see \thm{lower} in \sec{lower}), in the full Hilbert space, any quantum algorithm for generic time-dependent Hamiltonian simulations requires query complexity
\begin{align*}
\Omega\left(d\int_0^T\norm{H(\tau)}_{\max}\dd{\tau}+\frac{\log(1/\epsilon)}{\log\log(1/\epsilon)}\right).
\end{align*}

\begin{table}[!htbp]
\centering
\resizebox{1.0\columnwidth}{!}{
\renewcommand{\arraystretch}{2}
\begin{tabular}{ccc}
\hline System & Full Hilbert space~\cite{mizuta2024expliciterrorboundscommutator} & Low-energy subspace \\
\hline 
geometrically-local Hamiltonian & $O\left(\frac{N^{1/p}(gT)^{1+1/p}}{\epsilon^{1/p}}\right)$ & $\tilde{O}\left(g\frac{(\Delta+g\log(N/\epsilon))^{1/p}T^{1+1/p}}{\epsilon^{1/p}}\right)$ \\
power-law Hamiltonian & ${O}\left(\frac{N^{1/p}(gT)^{1+1/p}}{\epsilon^{1/p}}\log(N)\right)$ & $\tilde{O}\left(g\frac{(\Delta+g\log(N/\epsilon))^{1/p}T^{1+1/p}}{\epsilon^{1/p}}\right)$ \\
\hline
\end{tabular}}
\caption{A summary of low-energy simulation complexity (Trotter number) for different $k$-local Hamiltonians in the adiabatic regime where $T=\Omega(\tfrac{gN}{\Delta}{(\sigma\tfrac{\|{\dot{H}}\|}{\gamma^2}+\sigma\sqrt{\sigma} \tfrac{\|{\dot{H}}\|^2}{\gamma^3}+\sigma\tfrac{\|{\ddot{H}}\|}{\gamma^2})})$. 
For geometrically-local Hamiltonians, $g=O(1)$ can be omitted. 
For power-law Hamiltonians, $g=O(1),~ O(\log N)$, and $O(N^{1-\alpha / D})$ for $\alpha>D, \alpha=D$, and $\alpha<D$, where $D$ is the space dimension of the lattice. 
The improvement could be interpreted as replacing the full extensive energy scale by the effective low-energy scale in regimes where $\Delta+g\log(N/\epsilon)\ll gN$.}
\label{tab:result_table}
\end{table}

\vspace{-1em}
\subsection{Techniques}
We next outline the main idea of the proof ingredients.
In the following, we refer to the simulation error for each time step $\delta$ as the short-time simulation error, and the whole simulation error over $T$ as the long-time simulation error.

\paragraph{Short-time simulation error.} To derive an upper bound with commutator scaling for the short-time simulation error, we first recall the time-independent case~\cite{childs2021theory}. 
According to the variation-of-parameters formula, the $p$-th order time-independent product formula $S_p(\delta)$ and the exact evolution operator $V(\delta)=e^{-iH\delta}$ satisfy
\begin{align*}
S_p(\delta)-V(\delta)=\int_0^\delta\dd \tau e^{-iH(\delta-\tau)}\left(\frac{\dd}{\dd \tau}S_p(\tau)+iHS_p(\tau)\right)=\int_0^\delta\dd \tau e^{-iH(\delta-\tau)}S_p(\tau)\Delta(\tau),
\end{align*}
where $\Delta(\tau)=S_p^\dagger(\tau)(\frac{\dd}{\dd \tau}S_p(\tau)+iHS_p(\tau))$ can be written as the form $e^{\tau A_1} e^{\tau A_2}\cdots  e^{\tau A_q} B e^{-\tau A_q} \cdots e^{-\tau A_2}  e^{-\tau A_1}$ for some Hermitian operators $A_1,A_2,...,A_q$, and can be further expanded as nested commutators. 
Recently, Ref.~\cite{Mizuta_2025} significantly improved the analysis of the low-energy simulation error by directly appending the low-energy projector $\Pi$ at both sides of the above equation:
\begin{gather*}
(S_p(\delta)-V(\delta))\Pi =\int_0^\delta\dd \tau e^{-iH(\delta-\tau)}S_p(\tau)\Delta(\tau)\Pi,\\
\norm{(S_p(\delta)-V(\delta))\Pi} \leq\int_0^\delta\dd{\tau}\norm{\Delta(\tau)\Pi}.
\end{gather*}
The order condition indicates that $\Delta(\tau)=O(\tau^p)$, so we only need to reserve the remainder in its $(p-1)$-th order Taylor expansion.
However, in this way $\norm{\Delta(\tau)\Pi}$ takes the following form:
\begin{align*}
\norm{e^{iH_\gamma}\cdot \text{commutator}\cdot e^{-iH_\gamma}\cdot \Pi}\neq \norm{\text{commutator}\cdot \Pi}
\end{align*}
for some commutator, which forbids us to apply the projection lemmas to the nested commutator. 
As a simple illustrative example, we consider the $(p-1)$-th order Taylor expansion of $e^{\tau A}Be^{-\tau A}$:
\begin{align*}
e^{\tau A}Be^{-\tau A} = B+\tau\operatorname{ad}_AB+\frac{\tau^2}{2!}\operatorname{ad}_A^2B+\cdots + \frac{\tau^{p-1}}{(p-1)!}\operatorname{ad}_A^{p-1}B + \underbrace{\int_0^\tau \dd{\tau}_1 \frac{\tau_1^{p-1}}{(p-1)!} e^{(\tau-\tau_1)A}\operatorname{ad}_A^{p}Be^{-(\tau-\tau_1)A}}_{O(\norm{[A,\dots,[A,B]]}\cdot \tau^p)}.
\end{align*}
It can be observed that all the lower-order terms only contain commutators. 
However, the commutator in the integral remainder is sandwiched by $e^{\tau A}$ if we only expand the Taylor series to the $(p-1)$-th order.
To address this issue, we expand the Taylor series to higher orders to calculate the benefit to the leading-order term from simulating only states in the low-energy subspace.

In this work, we leverage the Floquet theory~\cite{mizuta2024expliciterrorboundscommutator} to calculate a time-dependent analogue to the residual term $\Delta(\tau)$ above. 
We append the low-energy projector to its high-order Taylor series, and compute upper bounds on different terms using a series of projection lemmas as tools.

\paragraph{Long-time simulation error.} To illustrate why we cannot simply decompose the long-time simulation error into sum of each step, we recall the proof of triangle inequality for telescoping operator products: 
Given unitaries $\{U_i\}_{i=1}^r$ and $\{V_i\}_{i=1}^r$, $\norm{U_1U_2\cdots U_r-V_1V_2\cdots V_r}\leq\sum_{i=1}^r\norm{U_i-V_i}$ since
\begin{align*}
\norm{U_1U_2\cdots U_r-V_1V_2\cdots V_r} &\leq \norm{U_1U_2\cdots U_r-U_1V_2\cdots V_r}+\norm{U_1V_2\cdots V_r-V_1V_2\cdots V_r}\\
& = \norm{U_1(U_2\cdots U_r-V_2\cdots V_r)}+ \norm{(U_1-V_1)V_2\cdots V_r} \\
& \leq \sum_{i=1}^r\norm{U_1\cdots U_{i-1}(U_i-V_i)V_{i+1}\cdots V_r}=\sum_{i=1}^r\norm{U_i-V_i}.
\end{align*}
However, if we append an additional operator such as the projector $\Pi$, we can only derive
\begin{align*}
    \norm{(U_1U_2\cdots U_r-V_1V_2\cdots V_r)\Pi} \leq \sum_{i=1}^r\norm{U_1\cdots U_{i-1}(U_i-V_i)V_{i+1}\cdots V_r\Pi}\neq \sum_{i=1}^r\norm{(U_i-V_i)\Pi}.
\end{align*}
In general, the last two expressions do not equal unless $[V_i,\Pi]=0$ or $[U_i,\Pi]=0$. 
On the other hand, note that in adiabatic perturbation theory, there are upper bounds for the form of $\norm{(I-\Pi)V\Pi}$. Therefore, intuitively, we can split the long-time simulation error into the sum of short-time simulation errors plus additional leakage errors, that is,
\begin{align*}
    \sum_{i=1}^r\norm{U_1\cdots U_{i-1}(U_i-V_i)V_{i+1}\cdots V_r\Pi} & = \sum_{i=1}^r\norm{U_1\cdots U_{i-1}(U_i-V_i)(I-\Pi+\Pi)V_{i+1}\cdots V_r\Pi}\\
    & \leq \sum_{i=1}^r\left(\norm{(U_i-V_i)\Pi}+\norm{(U_i-V_i)(I-\Pi)V_{i+1}\cdots V_r\Pi}\right)\\
    & \leq \sum_{i=1}^r\left(\norm{(U_i-V_i)\Pi}+\norm{U_i-V_i}\cdot\norm{(I-\Pi)V_{i+1}\cdots V_r\Pi}\right).
\end{align*}

\paragraph{Lower bound for query complexity.} Time-independent Hamiltonians can be regarded as a subset of time-dependent Hamiltonians. 
Therefore, any generic quantum algorithm for time-dependent Hamiltonian simulations also applies to time-independent Hamiltonian simulations, and thus cannot violate the existing lower bound concerning the time-independent case~\cite{Berry_2015}.

Among the parameters for sparse Hamiltonian simulations, we keep the sparsity $d$ unchanged and consider the time-dependent variant of the max-norm $\norm{H(t)}_{\max}$, such that the derived lower bound is compatible with the above arguments. We take a linear time-dependent Hamiltonian as the hard instance, and construct a reduction from its simulation to string parity computation, which has a known lower bound on both probabilistic Turing machine and quantum Turing machine~\cite{beals1998quantumlowerboundspolynomials,farhi1998limit}.

\subsection{Open questions}\label{sec:open}
Our paper leaves several open questions for future investigation:
\begin{itemize}
     \item \textbf{Interaction picture.} In this work, we focus on product formulas in the Schr\"{o}dinger picture. It is natural to ask whether the low-energy assumption can also enhance interaction picture Hamiltonian simulations~\cite{low2019hamiltoniansimulationinteractionpicture} as another frequently considered time-dependent scenario.

     \item \textbf{$L^{1}$-norm scaling.} Our error derivation is based on~\cite{mizuta2024expliciterrorboundscommutator} which bounds the nested commutator using max-norm. Is there possibility to improve our max-norm scaling to  $L^{1}$-norm scaling as the continuous qDRIFT algorithm and the rescaled Dyson series algorithm~\cite{Berry_2020}. 
     
    \item \textbf{Non-smooth Hamiltonian simulation.} The time-dependent product formulas require appropriate smooth conditions to achieve the desired error scaling. On the other hand, the current qubitization approaches~\cite{Mizuta_2023} for time-dependent Hamiltonians are also restricted to similar prerequisites. Inspired by this, can we prove a tighter lower bound for non-smooth Hamiltonian simulations, which may indicate that some algorithms are already optimal for generic cases?
\end{itemize}

\section{Preliminaries}\label{sec:pre}
This section presents the notation and modeling assumptions used throughout the proof. We first define the multi-linear local Hamiltonian input model and the time-dependent low-energy projector. We then explain the reasons why the usual fixed-threshold definition of a low-energy subspace is inconvenient for time-dependent spectra, and why a fixed number of instantaneous low-energy eigenstates is the more suitable assumption for adiabatic evolution. Finally, we introduce the projection and product-formula ingredients used in the error analysis.

\subsection{Notations}\label{sec:notation}
Throughout this work, we use the notation $\tilde{O}(\cdot)$, which omits the polylogarithmic dependence on the parameters. We summarize the notations in \tab{notations}.
\begin{table}[!htbp]
\centering
\resizebox{1.0\columnwidth}{!}{
{\begin{tabular}{ll|ll}
\hline
Symbol & Definition & Symbol & Definition \\ 
\hline
\hline
$A$ & Operator to apply projection & $R_A$ & Interaction strength on $A$ \\ 
$N$ & Number of spins & $T$ & Total evolution time\\
$\epsilon$ & Target simulation error & $r$ & Trotter number \\
$\Pi_{\sigma}$ & Spectral projector & $\delta$ & $T/r$, unit time step \\
$ \sigma $ & Number of eigenstates & $\gamma$ & Spectral gap between $\Pi_\sigma$ and $I-\Pi_\sigma$\\
$\Delta$ & Low-energy threshold & $\Delta'$ & $\geq\Delta$, effective low-energy norm \\
$s$ & $t/T\in[0,1]$, scaled time & $k$ & $k$-local Hamiltonian \\
$g$ & Interaction strength on a single spin & $\lambda$ &$(2gk)^{-1}$, abbreviation\\
$p$ & Order of product formula & $q$ & Terms number of product formula\\
$V$ & Exact evolution operator & $U$ & Unitary operator to approximate $V$\\ 
\hline
\end{tabular}}
}
\caption{{The notation table.}}
\label{tab:notations}
\end{table}

\subsection{Model and setup}\label{sec:model}
 In this work, we consider the following $k$-local \emph{multi-linear} Hamiltonian defined on lattice $\Lambda$:
\begin{align}\label{eq:multilinear}
H(t)=\sum_{X\subset\Lambda}f_X(t)h_X,
\end{align}
where the local interaction terms $h_X=0$ for $|X|>k$. In the implementation of product formulas, we organize as follows:
\begin{equation}\label{eq:split}
H(t)=\sum_{\gamma=1}^\Gamma  H_{\gamma}(t),
\end{equation}
where each $H_{\gamma}(t)$ can be efficiently exponentiated on a quantum computer, that is, the local interaction terms in it commute with each other. 
For instance, for a 1D spin chain with nearest-neighbor interaction, we can always split $H(t)=\sum_{i=1}^Nh_{i,i+1}(t)$ into $\Gamma=2$ terms $H(t)=H_1(t)+H_2(t)$ regardless of the system size $N$, where
\begin{align*}
    H_1(t)=\sum_{i;\mathrm{odd}}h_{i,i+1}(t),\quad H_2(t)=\sum_{i;\mathrm{even}}h_{i,i+1}(t).
\end{align*}
We will see later that the multi-linear structure~\eqref{eq:multilinear} is essential for implementing the time-ordered exponential of each $H_\gamma(t)$.

For an arbitrary time-dependent Hamiltonian $H(t)$, the spectrum is allowed to change suddenly and dramatically in such a way that the low-energy subspace at a time $t$ can cover the high-energy spectrum at the following time $t+\delta$. 
Therefore, additional assumptions on derivatives are necessary to acquire improvements over full Hilbert space simulations. We thus assume that all the coefficients $f_X(t)$ and their derivatives up to a certain order are $O(1)$-bounded over the time interval $[0, T]$. 
The primary application scenario is adiabatic state preparation, where $H(t)$ slowly changes from one time-independent Hamiltonian to another. 

The physical intuition behind the low-energy improvement is that the complexity of full Hilbert space Hamiltonian simulation typically depends on a global energy scale such as $\norm{H(t)}$ or the spectral width $E_{\max}(t)-E_{\min}(t)$. When the dynamics stay inside a low-energy subspace, one expects this global dependence to be replaceable by the bandwidth $\Delta_{\mathrm{phys}}$ of the relevant subspace, namely the interval $[E_{\min}(t),E_{\min}(t)+\Delta_{\mathrm{phys}}]$.

For $f_X(t)$ does not change sign over $[0,T]$, we absorb the possible minus sign into $h_X$ such that $f_X(t)\geq0$ holds. 
Let $E_{0,X}$ denote the smallest eigenvalue of $h_X$. 
We replace each local term by $h_X-E_{0,X}I$ when needed, so that all local summands are positive semi-definite. 
This local shift changes $H(t)$ only by a scalar term $-\sum_X f_X(t)E_{0,X}I$ and therefore changes the exact evolution only by a global phase. Physically, the relevant low-energy subspace is measured relative to the instantaneous bottom of the spectrum, say $[E_{\min}(t),E_{\min}(t)+\Delta_{\mathrm{phys}}]$; the global shift $H(t)-E_{\min}(t)I$ would put this band in $[0,\Delta_{\mathrm{phys}}]$ without changing the dynamics. In contrast, shifting each local term separately satisfies the requirements to invoke projection lemmas, but may over-shift the Hamiltonian to global ground energy $>0$. Consequently, the low-energy scale $\Delta$ entering our bounds may be slightly larger than $\Delta_{\text{phys}}$. For frustration-free or naturally positive decompositions this scale is aligned with the physical bandwidth, whereas in non-frustration-free cases it may include a local zero-point correction. 

To formalize our low-energy simulation problem, the most primary challenge is that, the eigenvectors of the time-dependent Hamiltonian $H(t)$ vary with time, making the corresponding projectors $\Pi_{\leq\Delta}(t)=\sum_{n:E_n(t)\leq\Delta} \ket{E_n(t)}\bra{E_n(t)}$ unfixed. 
We remark that in the time-independent case, constant energy threshold $\Delta$ is equivalent to a subspace spanned by fixed low-energy eigenstates, while in the time-dependent case $\Pi_{\leq\Delta}(t)$ may cover nothing if the ground energy $E_0(t)$ increases beyond the constant $\Delta$, making $\Pi_{\leq\Delta}(t)$ possibly not continuous and well-defined.

To address these issues, we alternatively define the spectral projector $\Pi_\sigma(t)$ onto a fixed number of low-energy eigenstates as \eqref{eq:projector}, and further assume that this part of the spectrum has no overlap with the remaining in $I-\Pi_\sigma(t)$. 
As illustrated in \fig{spectral}, $\Pi_{\sigma}(t)$ is continuous even with energy level crossings. 
The derivative of the projector $\dot{\Pi}_\sigma(t)$ and the spectral gap $\gamma(t)$ can also be introduced in a well-defined way as
\begin{align}
\Pi_\sigma(t)=\sum_{n\leq\sigma} \ket{E_n(t)}\bra{E_n(t)}.\label{eq:projector}
\end{align}

\begin{figure}[htbp]
    \centering
    \begin{subfigure}[ht]{.35\linewidth}
        \includegraphics[width=\linewidth]{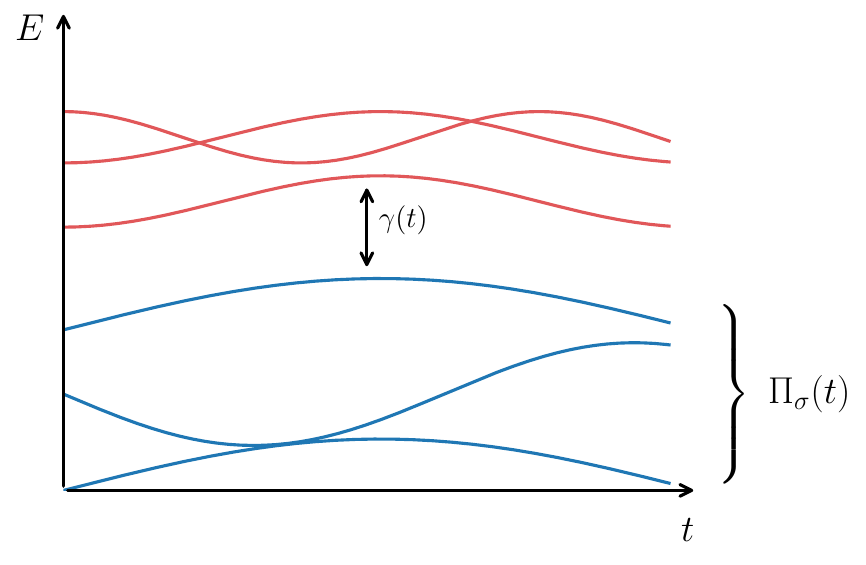}
        \label{fig:left}
    \end{subfigure}\hspace{5em}
    \begin{subfigure}[ht]{.35\linewidth}
        \includegraphics[width=\linewidth]{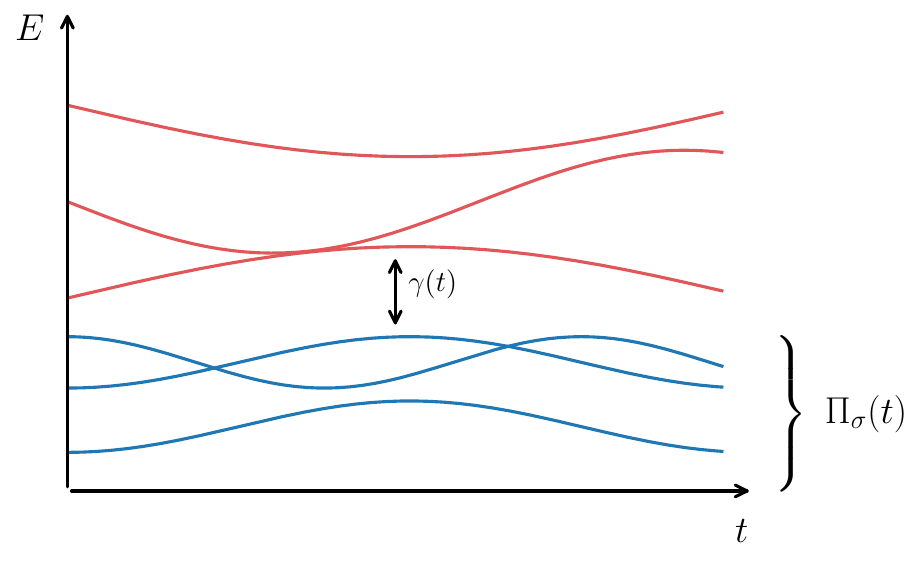}
        \label{fig:right}
    \end{subfigure}
    \caption{\textbf{Spectral Flow of Time-Dependent Hamiltonians.} We define the spectral projector onto the subspace spanned by the three low-energy eigenstates in blue, which has no overlap with the remaining eigenstates in red. In this illustrated case $\sigma=3$, $\Pi_{3}(t)=\ket{\psi_1(t)}\bra{\psi_1(t)}+\ket{\psi_2(t)}\bra{\psi_2(t)}+\ket{\psi_3(t)}\bra{\psi_3(t)}$.}
    \label{fig:spectral}
\end{figure}

In this work, we assume that the initial state is only supported on a low-energy subspace of the initial Hamiltonian $H(t=0)$ with the projector $\Pi_{\sigma}(t=0)$, which naturally fits with actual applications like adiabatic state preparation. 
The formal definition of our low-energy simulation problem is given in the following \defn{1}:

\begin{definition}\label{defn:1}
Given a Hamiltonian $H(t)$ in the form of \eqref{eq:multilinear} and a quantum state only supported on the low-energy subspaces of $H(t=0)$ spanned by fixed number of low-energy eigenstates, our goal is to find a quantum circuit $U$ such that the simulation error is below some threshold $\epsilon$, i.e., $\norm{(U-V)\Pi_{\sigma}(0)}\leq\epsilon$, where $V=\mathcal{T} e^{-i \int_0^T H(\tau) \dd{\tau}}$ is the target simulation over time $[0,T]$.
\end{definition}

Since the exact evolution operator $V$ with time-ordered exponential does not commute with the constructed projector $\Pi_\sigma$, we cannot take it for granted that
\begin{align}\label{longtime}
    \left\|(U_p(T,0)-V(T,0)) \Pi_{\sigma}(0)\right\|\overset{?}{\leq} r\left\|(U_p(\delta,0)-V(\delta,0)) \Pi_{\sigma}(0)\right\|.
\end{align}
To circumvent the issue in \eqref{longtime}, we introduce the following adiabatic perturbation theorem with multiple eigenstates to bound the leakage error when decomposing the simulation into time steps:

\begin{lemma}[Theorem 3 of~\cite{jansen_bounds_2007}] \label{lem:adiabatic} For a time-dependent Hamiltonian $H(t)$ over $t\in[0,T]$, let scaled time $s=\frac{t}{T}\in [0,1]$. Suppose that the spectral projector $\Pi_\sigma(s)$ contains $m(s)\leq\sigma$ eigenvalues (each possibly degenerate, crossings permitted) separated by a gap $\gamma(s)$ from the rest of the spectrum, then
\begin{align*}
\norm{(I-\Pi_\sigma(1))V(1,0)\Pi_\sigma(0)}\leq \frac{1}{T}\left(\frac{\sigma\|{\dot{H}(0)}\|}{\gamma(0)^2}+\frac{\sigma\|{\dot{H}(1)}\|}{\gamma(1)^2}+\int_0^1\dd{s}\left(\frac{\sigma\|\ddot{H}(s)\|}{\gamma(s)^2}+7 \sigma \sqrt{\sigma} \frac{\|\dot{H}(s)\|^2}{\gamma(s)^3}\right)\right).
\end{align*}
\end{lemma}

\subsection{Projection lemmas}~\label{sec:lemma}
To analyze the suppression of the operator norm in the subspace of time-dependent Hamiltonians, we first prove a fine-tuned  version of the backbone lemma in~\cite{arad2016connecting}:

\begin{lemma}[Time-dependent version of Theorem 2.1 of~\cite{arad2016connecting}]\label{lem:LeakOp} For a $k$-local time-dependent Hamiltonian $H(t)$ as \eqref{eq:multilinear}, we shift all the local interaction terms to positive semi-definite. 
At an instantaneous time $t$, we can regard $H(t)$ as a time-independent Hamiltonian with a fixed spectrum, the projector defined in \eqref{eq:projector} is equivalent to $\Pi_\sigma(t)=\Pi_{\leq\Delta_t}$ for some energy threshold depending on $t$.

Given an operator $A$, we denote $R$ the sum of strengths of $f_X(t)h_X$ that do not commute with $A$. Given energy threshold $\Lambda_t'$ and $\Lambda_t$, we have
\begin{align*}
\norm{\Pi_{>\Lambda_t'} A\Pi_{\leq\Lambda_t}}\leq\norm{A}\cdot e^{-\lambda(\Lambda_t'-\Lambda_t-2R)},
\end{align*}
where $\lambda=(2gk)^{-1}$, $g$ is an upper bound of the sum of strength acting on a single spin at $t$. \end{lemma}

\begin{proof}
For any $0\leq s<{(gk)}^{-1}$, according to Lemma 3.1 of Ref.~\cite{arad2016connecting}, 
    \begin{align*}
        \norm{\Pi_{>\Lambda_t'} A\Pi_{\leq\Lambda_t}}&=\norm{\Pi_{>\Lambda_t'} e^{-sH(t)}e^{sH(t)}A e^{-sH(t)}e^{sH(t)}\Pi_{\leq\Lambda_t}}\\
        &\leq \norm{e^{sH(t)}Ae^{-sH(t)}}\cdot e^{-s(\Lambda_t'-\Lambda_t)}\\
        &\leq  \frac{e^{-s(\Lambda_t'-\Lambda_t)}}{(1-sgk)^{R/gk}}.
    \end{align*}
To minimize the RHS, by simple calculus we pick $s=\frac{1}{gk}[1-\frac{R}{\Lambda_t'-\Lambda_t}]$, then
\begin{align*}
    \norm{\Pi_{>\Lambda_t'} A\Pi_{\leq\Lambda_t}}\leq\norm{A}\cdot e^{-\frac{1}{2gk}(\Lambda_t'-\Lambda_t-2R)}.
\end{align*}
\end{proof}

In addition, we present the behavior of the nested commutator in the low-energy subspace:
\begin{lemma}[Theorem S1 and Corollary S2 of~\cite{Mizuta_2025}] \label{lem:lowcommutator}
Consider $(p+1)$ groups of positive semi-definite local interaction terms $h_X^{p+1},\dots, h_X^{1}\geq0,~ \forall X\subset\Lambda$, the induced $1$-norm of each group is jointly upper bounded by $g$. 
If the projector $\Pi$ satisfies $\norm{\Pi \sum_X h_X^{j}\Pi}\leq \Delta$ for each group index $j=1,\dots,p+1$, the nested commutator with different Hamiltonians in each layer satisfies:
\begin{align*}
    \sum_{X_{p+1},\dots ,X_1}\norm{\Pi[h_{X_{p+1}}^{p+1},\dots,[h_{X_2}^{2},h_{X_1}^{1}]]\Pi}\leq p!(2kg)^p\Delta,
\end{align*}
while in the full Hilbert space the sum of norms is bounded by substituting $\Delta$ with $gN$:  
\begin{align*}
  \sum_{X_{p+1},\dots ,X_1}\norm{[h_{X_{p+1}}^{p+1},\dots,[h_{X_2}^{2},h_{X_1}^{1}]]}\leq p!(2kg)^pgN.  
\end{align*}
\end{lemma}

\subsection{Product formula for time-dependent Hamiltonian simulation}
There are essentially two types of product formulas for time-dependent Hamiltonian simulations, which employ different exponential terms. 
Huyghebaert and De Raedt first showed how to generalize the Trotter formula to ordered exponentials~\cite{JHuyghebaert_1990}, where the time integral over exponent and the time-ordering operator are reserved. We thus call this approach generalized product formulas and denote them by $U_p$ in this paper.
Subsequently, Suzuki developed another decomposition scheme by merely substituting time-independent product formulas with midpoint times~\cite{1993161}, which we call standard product formulas and denote them by $S_p$ in this paper. Besides, there is research~\cite{Ikeda_2023} exploiting time-dependent product formulas with the minimum number of exponential terms. Furthermore, Ref.~\cite{fkh5-b669} provides a unifying framework for all these time-dependent product formulas.

Given $H(t)=\sum_{\gamma=1}^\Gamma H_{\gamma}(t)$, to approximate the exact evolution operator $V(T,0)$, we first decompose $T$ into equidistant steps $\delta=T/r$. The first- and second-order generalized product formulas are defined as
\begin{align*}
& U_1(t+\delta,t)=\prod_\gamma^{\leftarrow} \mathcal{T} \left[e^{-i \int_t^{t+\delta} H_{\gamma}(\tau) \dd{\tau}}\right], \\
& U_2(t+\delta,t)=\prod_\gamma^{\rightarrow} \mathcal{T} \left[e^{-i \int_{t+\delta/2}^t H_{\gamma}(\tau) \dd{\tau}}\right] \cdot \prod_\gamma^{\leftarrow} \mathcal{T} \left[e^{-i \int_t^{t+\delta/2} H_{\gamma}(\tau) \dd{\tau}}\right],
\end{align*}
while the standard product formula is defined by the midpoint rule as
\begin{align*}
    S_2(t+\delta,t) = \prod_\gamma^{\rightarrow} e^{-iH_{\gamma}(t+\delta/2)\delta}\cdot \prod_\gamma^{\leftarrow} e^{-iH_{\gamma}(t+\delta/2)\delta},
\end{align*}
where each (time-ordered) exponential term can be easily implemented by local quantum gates. 

Ref.~\cite{Wiebe_2010} formalized the higher-order standard product formulas and proved their error scaling, given the order of differentiability. Specifically,
\begin{align}\label{eq:recursion}
\begin{split}
    S_{2p}(t+\delta, t) & =  S_{2p-2}\left(t+\delta, t+\left[1-u_p\right] \delta \right) S_{2p-2}\left(t+\left[1-u_p\right] \delta, t+\left[1-2 u_p\right] \delta \right) \\
& \times S_{2p-2}\left(t + \left[1-2 u_p\right] \delta, t +2 u_p \delta \right) S_{2p-2}\left(t +2 u_p \delta, t + u_p \delta \right) S_{2p-2}\left(t+u_p \delta, t\right),
\end{split}
\end{align}
where $u_p = (4-4^{1/(2p-1)})^{-1}$. Although this recursive method in \cite{Wiebe_2010} is developed for standard product formulas, it can be verified that the vanishing of lower-order error terms still holds for generalized product formulas, given appropriate smooth conditions.
Given each $H_\gamma(t)\in C^{p}$, that is, the $0,1,\ldots, p$-th order derivatives exist and are continuous, then
\begin{align*}
U_p(t+\delta,t)= \mathcal{T} \left[e^{-i \int_t^{t+\delta} H(\tau) \dd{\tau}}\right]+O(\delta^{p+1}),\\
S_p(t+\delta,t)= \mathcal{T} \left[e^{-i \int_t^{t+\delta} H(\tau) \dd{\tau}}\right]+O(\delta^{p+1}).
\end{align*}

Recently, Ref.~\cite{mizuta2024expliciterrorboundscommutator} derives explicit error bounds with commutator scaling for time-dependent product formulas. We review their Floquet theory details and list the important results as lemmas in \append{floquet}. Given each $H_{\gamma}(t)\in C^{p+2}$, they first embed each $H_{\gamma}(t)$ into a time-periodic Hamiltonian with $T^{\mathrm{ex}}=2T$ using a bump function. 
Then, by employing Floquet theory, they map the time-periodic Hamiltonian into a time-independent Hamiltonian on infinite-dimensional space, whose Trotter error can be bounded with commutator scaling for time-independent scenarios~\cite{childs2021theory}. 
As a result, they obtain a meaningful convergent bound back to the original Hilbert space as follows:
\begin{equation*}
\|U_p(t+\delta,t)-V(t+\delta,t)\| \in {O}\left(\max _{\tau \in[t, t+\delta]} \alpha_{\mathrm{com}}^{p+1}(\tau) \cdot \delta^{p+1}\right),
\end{equation*}
where the nested commutator factor $\alpha_{\mathrm{com}}^{p+1}(t)$ is given by
\begin{equation}\label{eq:nested1}
\alpha_{\mathrm{com}}^{p+1}(t)=\sum_{\gamma_1, \dots, \gamma_p=1}^{\Gamma+1} \sum_{\gamma_{0}=1}^{\Gamma}\left\|\left\{\prod_{j=1}^p \mathcal{D}_{\gamma_j}(t)\right\} H_{\gamma_{0}}(t)\right\|, \quad \mathcal{D}_\gamma(t)= \begin{cases}\operatorname{ad}_{H_{\gamma}(t)} & (\gamma=1, \dots, \Gamma) \\ 2 \Gamma \frac{\mathrm{~d}}{\mathrm{~d} t} & (\gamma=\Gamma+1)\end{cases}.
\end{equation}

Given each $H_{\gamma}(t)\in C^{p+2}$,
Ref.~\cite{mizuta2024expliciterrorboundscommutator} also proves that for standard product formulas,
\begin{equation*}
\|S_p(t+\delta,t)-V(t+\delta,t)\| \in {O}\left(\max _{\tau \in[t, t+\delta]} \overline{\alpha}_{\mathrm{com}}^{p+1}(\tau) \cdot \delta^{p+1}\right),
\end{equation*}
where the alternative nested commutator factor $\overline{\alpha}_{\mathrm{com}}^{p+1}(t)$ is given by
\begin{equation*}
\overline{\alpha}_{\mathrm{com}}^{p+1}(t)=\sum_{\gamma_1, \dots, \gamma_p=1}^{\Gamma+1} \sum_{\gamma_{0}=1}^{\Gamma}\left\|\left\{\prod_{j=1}^p \mathcal{\overline{D}}_{\gamma_j}(t)\right\} H_{\gamma_{0}}(t)\right\|, \quad \mathcal{\overline{D}}_\gamma(t)= \begin{cases}\operatorname{ad}_{H_{\gamma}(t)} & (\gamma=1, \dots, \Gamma) \\ \Gamma\frac{\mathrm{~d}}{\mathrm{~d} t} & (\gamma=\Gamma+1)\end{cases}.
\end{equation*}

\section{Error Analysis of Time-Dependent Low-Energy Simulation}\label{sec:analysis}

We now prove the claimed Trotter number under the assumptions stated above. The proof separates the genuinely low-energy short-time product-formula error from the leakage accumulated by the exact evolution between different instantaneous low-energy subspaces. In adiabatic state preparation, the Hamiltonian is typically parameterized by a smooth interpolation path of the form
\begin{equation}\label{eq:adiabatic_interpolation}
    H(t)=f_1\left(\frac{t}{T}\right)H^{(1)}+f_2
    \left(\frac{t}{T}\right)H^{(2)},
\end{equation}
where $s=t/T\in[0,1]$ is the scaled time. We take the real-valued schedules $f_1,f_2\in C^{p+2}([0,1],\mathbb{R})$ and impose the endpoint conditions
\begin{equation*}
    f_1(0)=1,\quad f_2(0)=0,\qquad
    f_1(1)=0,\quad f_2(1)=1,
\end{equation*}
\nopagebreak[4]
so that $H(0)=H^{(1)}$ and $H(T)=H^{(2)}$. Adiabaticity additionally requires a spectral gap above the relevant low-energy sector along the path.

Since the proof below repeatedly uses functions of derivatives at different orders of the Hamiltonian, we state the two structural assumptions at the beginning of the section. The first \assum{1} is an ordinary smoothness condition on the coefficient functions. The second, \assum{2}, is the main low-energy projected derivative condition used to apply \lem{lowcommutator} to the nested commutators.

\begin{assumption}\label{assum:1}
    We assume the smoothness of each coefficient function $f_X$ such that 
    \begin{equation}\label{eq:smooth}
    \max _{t \in[0, T]}\abs{\dv[n]{t}f_X(t)} = {O}(1),~\forall n=0,1, \ldots, p+2.
\end{equation}
This indicates that the coefficient functions and their derivatives are uniformly bounded by constants independent of $N$, $\epsilon$, and a linear dependence on $T$. For instance, $f_X(t)=\sin\left(\frac{t}{T}\right)$ or $\frac{1}{|\operatorname{diam(X)}|^2}$.
\end{assumption}

\begin{assumption}\label{assum:2} 
For a projector $\Pi(t)$ onto the low-energy spectrum of $H(t)$ with energy threshold $\Delta_t'$ at any instantaneous time $t\in[0,T]$, by definition $\left\|\Pi(t)\sum\nolimits_X |f_X(t)|h_X \Pi(t)\right\|=\left\|\Pi(t) H(t) \Pi(t)\right\|\leq \Delta_t'$. We further assume that the $n$-th order derivatives of $H(t)$ also satisfy
\begin{align*}
\left\|\Pi(t)\sum\nolimits_X |f_X^{(n)}(t)|h_X \Pi(t)\right\| \leq\Delta_t',\quad n=1,\dots,p,
\end{align*}
where the inequality can be relaxed to $O(\Delta_t')$. 
\end{assumption}

\noindent \assum{2} is an additional projected derivative condition, which is not a consequence of locality alone and may fail for a rapidly varying path. A conservative sufficient condition is obtained by bounding the corresponding derivative in the full Hilbert space. If, for every $t\in[0,T]$ and $n=1,\dots,p$,
\begin{align*}
\left\|\sum\nolimits_X |f_X^{(n)}(t)|h_X\right\|=O(\Delta_t'),
\end{align*}
then \assum{2} follows immediately from $\|\Pi A\Pi\|\leq\|A\|$. For an adiabatic interpolation written in the scaled time, $f_X(t)=a_X(s)$ with $s=t/T$, the chain rule gives $f_X^{(n)}(t)=T^{-n}\partial_s^n a_X(s)$. Hence, if the dimensionless path has full-space derivative of the natural extensive size
\begin{align*}
\left\|\sum\nolimits_X |\partial_s^n a_X(s)|h_X\right\|=O(gN),\qquad n=1,\dots,p,
\end{align*}
then
\begin{align*}
\left\|\Pi(t)\sum\nolimits_X |f_X^{(n)}(t)|h_X\Pi(t)\right\|\leq \left\|\sum\nolimits_X |f_X^{(n)}(t)|h_X\right\|=O\left(\frac{gN}{T^n}\right).
\end{align*}
Therefore \assum{2} holds, for example, whenever $gN/T^n=O(\Delta_t')$ for the derivative orders used in the proof. This full-space check is only a sufficient condition; one may also verify the projected inequality in \assum{2} directly, which can be less restrictive. For example, in certain non-equilibrium quantum many-body dynamics, most of the $f_X$ is static, and only a few of them vary with time; then \assum{2} may extend to general non-adiabatic regimes. The spectral gap used later controls leakage between $\Pi_\sigma(t)$ and its complement, whereas \assum{2} controls the size of derivative layers after projection onto the selected low-energy sector.

Two standard interpolation paths for which this full-space check is explicit are the linear and trigonometric schedules
\begin{align*}
    H(t)&=\left(1-\frac{t}{T}\right)H^{(1)}+\left(\frac{t}{T}\right)H^{(2)},\\
    H(t)&=\cos{\left(\frac{\pi t}{2T}\right)}H^{(1)}+\sin{\left(\frac{\pi t}{2T}\right)}H^{(2)}.
\end{align*}
For the linear schedule, only the first derivative with respect to the scaled time $s=t/T$ is nonzero; for the trigonometric schedule, every fixed-order derivative with respect to $s$ is bounded by a constant. Consequently, the corresponding full-space derivative norm is $O(gN/T^n)$ for both schedules. Thus, whenever $gN/T^n=O(\Delta_t')$ for $n=1,\dots,p$, the full-space sufficient condition is satisfied and \assum{2} follows.

In \sec{scaling}, we derive a fine-grained analysis of the nested commutator of our input model based on \assum{1}. In \sec{shorttime}, we then leverage the variation-of-parameters formula and Floquet theory to expand the short-time simulation error into higher-order Taylor series containing nested commutators. We further apply projection lemmas to derive a low-energy simulation error with commutator scaling based on \assum{2}. Finally, in \sec{longtime} we employ the adiabatic perturbation theory with multiple eigenstates to bound the leakage error when decomposing the whole simulation into time steps. 

\subsection{Nested commutator of $k$-local multi-linear Hamiltonian}\label{sec:scaling}
In the time-independent scenario, operator commutation $[A,B]=0$ implies $e^{A+B}=e^{A}e^{B}$. 
However, regarding the time-ordered exponential $\mathcal{T}e^{A(t)+B(t)}$, in addition to $[A(t),B(t)]=0$, we further require $[A(t),A(t')]=[B(t),B(t')]=0$ to ensure the split, which can be verified by the definition. 
Therefore, with the multi-linear structure~\eqref{eq:multilinear}, the time-ordered exponential satisfies:
\begin{align*}
\mathcal{T}\exp\left(-i\int H_\gamma(\tau)\dd{\tau}\right) = \prod_{h_X\in H_{\gamma}}\mathcal{T}\exp\left(-i\int f_X(\tau)h_X\dd{\tau}\right) = \prod_{h_X\in H_{\gamma}}\exp(-i\int f_X(\tau)\dd{\tau}\cdot h_X),
\end{align*}
where the first equation is due to the fact that $[f_X(t)h_X,f_{X'}(t')h_{X'}]=0$ always holds even if $t'\neq t$ for each individual $H_\gamma$ as we assumed in Section~\sec{model},
and the second equation shows that the time-ordering operator before each local exponential is reduced, explaining why $H_{\gamma}(t)$ can be efficiently exponentiated on a quantum computer.

Appendix E of~\cite{childs2021theory} has shown that the nested commutator of local interaction terms satisfies:
\begin{equation}\label{eq:99}
\sum_{X_1, \dots, X_{p+1}}\norm{\left[h_{X_{p+1}}, \dots,\left[h_{X_2}, h_{X_1}\right]\right]}\leq p!(2k\|\hspace{-.1em}|h|\hspace{-.1em}\|_1)^p\cdot \sum_X\norm{h_X}=O(\|\hspace{-.1em}|h|\hspace{-.1em}\|_1^{p+1}N),
\end{equation}

We then focus on the splitting and scaling of the nested commutator~\eqref{eq:nested1}:
\begin{align*}
\alpha_{\mathrm{com}}^{p+1}(t) &=  \sum_{j=1}^{p+1}\sum_{\gamma_1,\dots,\gamma_j=1}^\Gamma(2\Gamma)^{p+1-j}\sum_{n_1+\cdots +n_j=p+1-j}\norm{\frac{\dd^{n_j}}{\dd{t}^{n_j}}[H_{\gamma_j}(t),\dots,\frac{\dd^{n_2}}{\dd{t}^{n_2}}[H_{\gamma_2}(t),\frac{\dd^{n_1}}{\dd{t}^{n_1}}H_{\gamma_1}(t)]]}\\
& \leq \sum_{j=1}^{p+1}\sum_{X_1,\dots,X_{j}}(2\Gamma)^{p+1-j}\sum_{n_1+\cdots +n_j=p+1-j}\norm{\frac{\dd^{n_j}}{\dd{t}^{n_j}}[f_{X_j}(t)h_{X_j},\dots,\frac{\dd^{n_2}}{\dd{t}^{n_2}}[f_{X_2}(t)h_{X_2},\frac{\dd^{n_1}}{\dd{t}^{n_1}}f_{X_1}(t)h_{X_1}]]}\\
& = \sum_{j=1}^{p+1}(2\Gamma)^{p+1-j}\sum_{X_1,\dots,X_{j}}\underbrace{\sum_{n_1+\cdots +n_j=p+1-j}\left|\frac{\dd^{n_j}}{\dd{t}^{n_j}}(f_{X_j}(t),\dots,\frac{\dd^{n_1}}{\dd{t}^{n_1}}f_{X_1}(t))\right|}_{S(p+1,j) \text{ summands } \leq \frac{j^{p+1}}{j!} F^j}\norm{[h_{X_j},\dots,[h_{X_2},h_{X_1}]]},
\end{align*}
where the first equation follows that each summand contains a $j$-layer nested commutator interleaved with $(p+1-j)$-times derivatives, the second line follows that the summation over each $H_{\gamma}$ iterates over all $h_X$, and the last line extracts all the coefficients in the nested commutator.

Given one of the total $C_{j-1}^p$ partitions $n_1+\cdots +n_j=p+1-j$, coefficient $\frac{\dd^{n_j}}{\dd{t}^{n_j}}(f_{X_j}(t),\dots\frac{\dd^{n_1}}{\dd{t}^{n_1}}f_{X_1}(t))$ contains $1^{n_1}2^{n_2}\cdots j^{n_j}$ summands, each upper bounded by $F^j$ where
\begin{align*}
F = \max_{n=0,1,\dots,p}\max_{X}\max_{t\in[0,T]}\left|\frac{\dd^n}{\dd{t}^n}f_X(t)\right|.
\end{align*}
\assum{1} implies $F=O(1)$, we then set the parameter $g$ as follows:
\begin{align}\label{eq:g}
    g=F\cdot \|\hspace{-.1em}|h|\hspace{-.1em}\|_1,
\end{align}
which is an upper bound of interaction strength on a single spin for both $H(t)$ and its derivatives at any time $t\in[0,T]$, meeting the requirement in \lem{LeakOp}.

There are in total $\sum_{n_1+n_2+\cdots n_j=p+1-j}1^{n_1}2^{n_2}\cdots j^{n_j}=S(p+1,j)\leq \frac{j^{p+1}}{j!}$ summands, where $S(\cdot,\cdot)$ denotes Stirling numbers of the second kind. Actually, $S(p+1,j)$ counts partitions of $p+1$ labeled derivative slots into $j$ nonempty ordered layers up to permutations inside each layer. The elementary bound follows from the standard estimate $S(n,j)\leq j^n/j!$, obtained by counting all functions from an $n$-element set to a $j$-element set and then dividing by $j!$. The nested commutator is thus bounded by
\begin{align*}
\alpha_{\mathrm{com}}^{p+1}(t) 
& \leq \sum_{j=1}^{p+1}(2\Gamma)^{p+1-j}\sum_{X_1,\dots,X_{j}}\frac{j^{p+1}}{j!}F^j\norm{[h_{X_j},\dots,[h_{X_2},h_{X_1}]]}\\
& \leq \sum_{j=1}^{p+1}(2\Gamma)^{p+1-j}j^p \cdot (2kg)^{j-1}gN\\
& \leq (2k(p+1))^{p}\cdot \left(\sum_{j=1}^{p+1}\Gamma^{p+1-j}g^{j}N\right).
\end{align*}

We will frequently refer to this form later. In the same way for the standard product formula,
\begin{align*}
\overline{\alpha}_{\mathrm{com}}^{p+1}(t)\leq (2k(p+1))^{p}\cdot \left(\sum_{j=1}^{p+1}\Gamma^{p+1-j}g^{j}N\right).
\end{align*}

\subsection{Short-time simulation error}\label{sec:shorttime}
We first bound the low-energy simulation error for each time step $\delta$. Concerning time-dependent product formulas, \lem{representation} in \append{floquet} presents the following error representation:
\begin{align}\label{eq:representation}
U_p(\delta, 0)-V(\delta, 0)= i \int_0^\delta \mathrm{~d} \tau V(\delta, \tau) U_p(\tau, 0)\left(\sum_{l \in \mathbb{Z}}\langle l| \Delta^F(\tau)|0\rangle\right),
\end{align}
where $\Delta^F$ is the residual term \eqref{eq:residual} defined for the mapped infinite-dimensional time-independent $H^F$. 
Appending a projector at both sides of the above variation-of-parameters equation~\cite{Mizuta_2025}, we can express the low-energy simulation error as follows:
\begin{gather*}
(U_p(\delta, 0)-V(\delta, 0))\Pi_{\sigma}(0)= i \int_0^\delta \mathrm{~d} \tau V(\delta, \tau) U_p(\tau, 0)\left(\sum_{l \in \mathbb{Z}}\langle l| \Delta^F(\tau)|0\rangle\right)\Pi_{\sigma}(0),\\
\norm{(U_p(\delta, 0)-V(\delta, 0))\Pi_{\sigma}(0)}\leq \int_0^\delta \mathrm{~d} \tau \norm{\sum_{l \in \mathbb{Z}}\langle l| \Delta^F(\tau)|0\rangle\Pi_{\sigma}(0)}.
\end{gather*}

For simplicity, throughout this work, we rewrite the generalized product formula as
\begin{align}
U_p(\delta,0)= \prod_{j=1,2, \dots, q}^{\leftarrow} U_{\gamma_j}\left(\beta_j \delta+\alpha_j \delta, \beta_j \delta\right),\quad\text{where}~U_{\gamma}(t',t)=\mathcal{T}\exp({\int_t^{t'} H_{_{\gamma}}(\tau)\dd{\tau}}).\label{eq:GPF}
\end{align}
The coefficients $\{\alpha_j\}$ follow the same coefficients of Suzuki-Trotter formulas in the time-independent case, and $\{\beta_j\}$ denotes the starting point of the time evolution of each Hamiltonian term in the product formula. For the $p$-th order Trotter-Suzuki formula defined by \eqref{eq:recursion}, there are in total $q=2\cdot5^{\lceil p/2\rceil-1}\cdot\Gamma$ exponential terms, and we have
\begin{align*}
    & \beta_1=0,\quad \beta_q+\alpha_q=1,\\
    & 0\leq\beta_j\leq1,\quad 0\leq\beta_j+\alpha_j\leq1.
\end{align*}

In the derivation of nested commutators from $\Delta(\tau)$~\cite{childs2021theory, mizuta2024expliciterrorboundscommutator}, they actually consider the $(p-1)$-th order Taylor series with integral remainders. 
Then the order condition of Trotter error indicates that only the $O(\tau^p)$ remainder reserves. 
By recursion, they express $\sum_{l \in \mathbb{Z}}\langle l| \Delta^F(\tau)|0\rangle$ as follows:
\begin{multline}\label{eq:remainder}
\hspace*{-2em}\sum_{l \in \mathbb{Z}}\langle l| \Delta^F(\tau)|0\rangle = \sum_{\gamma=1}^\Gamma\sum_{j'=1}^{2q-1}\sum_{\substack{n_1+\cdots+n_{j'}=p \\ n_{j^{\prime}} \neq 0}}U_{< \lceil j'/2 \rceil}(\tau,0)^\dagger\int_0^\tau \dd{\tau}_1\frac{i^{p}\left(\tau-\tau_1\right)^{n_{j^{\prime}}-1} \tau^{p-n_{j^{\prime}}}}{\left(n_{j^{\prime}}-1\right)!n_{j^{\prime}-1}!\cdots n_{1}!} C_{j'\rightarrow 2q-1}(H_{\gamma})\cdot  U_{< \lceil j'/2 \rceil}(\tau,0)\\
- \sum_{j=1}^q(\beta_j+\alpha_j)\sum_{j'=1}^{2j-1}\sum_{\substack{n_1+\cdots+n_{j'}=p \\ n_{j^{\prime}} \neq 0}}U_{< \lceil j'/2 \rceil}(\tau,0)^\dagger\int_0^\tau \dd{\tau}_1\frac{i^{p}\left(\tau-\tau_1\right)^{n_{j^{\prime}}-1} \tau^{p-n_{j^{\prime}}}}{\left(n_{j^{\prime}}-1\right)!n_{j^{\prime}-1}!\cdots n_{1}!} C_{j'\rightarrow 2j-1}(H_{\gamma_j})\cdot  U_{< \lceil j'/2 \rceil}(\tau,0)\\
+ \sum_{j=1}^q\beta_j\sum_{j'=1}^{2j-2}\sum_{\substack{n_1+\cdots+n_{j'}=p \\ n_{j^{\prime}} \neq 0}}U_{< \lceil j'/2 \rceil}(\tau,0)^\dagger\int_0^\tau \dd{\tau}_1\frac{i^{p}\left(\tau-\tau_1\right)^{n_{j^{\prime}}-1} \tau^{p-n_{j^{\prime}}}}{\left(n_{j^{\prime}}-1\right)!n_{j^{\prime}-1}!\cdots n_{1}!} C_{j'\rightarrow 2j-2}(H_{\gamma_j})\cdot  U_{< \lceil j'/2 \rceil}(\tau,0),
\end{multline}
where $U_{< \lceil j'/2 \rceil}$ denotes truncating the product formula~\eqref{eq:GPF} from $j=1,\dots,q$ to $1, \dots,\lceil j'/2 \rceil-1$. The commutators are taken at some time $\tau'\in[0,\tau]$ decided by $\tau$ and $\tau_1$ in the following form:

\begin{gather*}
C_{j'\rightarrow 2q-1}(A) = U_{\gamma_{\lceil j'/2 \rceil}}(\alpha_{\lceil j'/2 \rceil}\tau_1)^\dagger\underbrace{\left[ (\tilde{\alpha}_{2q-j'}D_{2q-j'}(\tau'))^{n_{j'}}\cdots (\tilde{\alpha}_{2q-1}D_{2q-1}(\tau'))^{n_{1}}A(\tau') \right]}_{\text{commutator}} U_{\gamma_{\lceil j'/2 \rceil}}(\alpha_{\lceil j'/2 \rceil}\tau_1)\\
\tilde{\alpha}_{j'}=\begin{cases}\alpha_j & (j^{\prime}=2 j-1) \\ \beta_{j+1}-\beta_j -\alpha_j & (j^{\prime}=2 j)\end{cases},\quad D_{j'}(t)=\begin{cases}\operatorname{ad}_{H_{\gamma_j}(t)}+i\frac{\dd}{\dd{t}} & (j^{\prime}=2 j-1) \\ i\frac{\dd}{\dd{t}} & (j^{\prime}=2 j)\end{cases}.
\end{gather*}

In the full Hilbert space, the norm of each unitary $U_\gamma$ equals $1$. Then Ref.~\cite{mizuta2024expliciterrorboundscommutator} derives
\begin{align*}
\left\|\sum_{l \in \mathbb{Z}}\langle l| \Delta^F(\tau)|0\rangle\right\|& \leq 3p\tau^p\cdot \max_{\tau_1\in[0,\tau]}\sum_{j=1}^q \sum_{j_1',\dots,j_{p}'=1}^{2q-1} \norm{(\tilde{\alpha}_{j_1'}D_{j_1'}(\tau_1))\cdots(\tilde{\alpha}_{j_p'}D_{j_p'}(\tau_1))H_{\gamma_j}(\tau_1)}\\
& \leq 3\left(\frac{q}{\Gamma}\right)^{p+1}  \cdot \max_{\tau_1\in[0,\tau]} \alpha_{\mathrm{com}}^{p+1}(\tau_1)\cdot p\tau^p.
\end{align*}

In the low-energy subspace, issue is that $\norm{\sum_{l \in \mathbb{Z}}\langle l| \Delta^F(\tau)|0\rangle\Pi_\sigma(0)}$ takes the following form:
\begin{align*}
\norm{ U_{\gamma}^\dagger\cdot \text{commutator}\cdot U_{\gamma}\cdot \Pi_{\sigma}(0)}\neq \norm{\text{commutator}\cdot \Pi_{\sigma}(0)},
\end{align*}
which forbids us to apply the projection lemmas to the nested commutator since it is sandwiched by some $U_{\gamma}$. To address this issue, we expand the Taylor series to higher orders to calculate the benefit to the leading-order term from simulating only states in the low-energy subspace. The $(p_0-1)$-th order Taylor series with integral remainders takes the following form:
\begin{align}
\sum_{l \in \mathbb{Z}}\langle l| \Delta^F(\tau)|0\rangle = \sum_{n=p}^{p_0-1}A_n \tau^n + A_{p_0}(\tau),\quad A_{p_0}(\tau)\sim O(\tau^{p_0}),\label{eq:expansion}
\end{align}
where $p_0>p$, $A_n$ only contains commutator. 
The order condition ensures that $\tau^0,\dots,\tau^{p-1}$ terms cancel. 
We then bound the short-time simulation error by
\begin{align}\label{eq:form}
\int_0^\delta \mathrm{~d} \tau \norm{\sum_{l \in \mathbb{Z}}\langle l| \Delta^F(\tau)|0\rangle\Pi_{\sigma}(0)} & \leq \sum_{n=p}^{p_0-1} \frac{\delta^{n+1}}{n+1}\norm{A_n\Pi_{\sigma}(0)}+\underbrace{\int_0^\delta \dd{\tau} \norm{A_{p_0}(\tau)\Pi_{\sigma}(0)}}_{\leq \int_0^\delta \dd{\tau} \norm{A_{p_0}(\tau)}}.
\end{align}

In \append{floquet}, we calculate the Taylor expansion of $\Delta^F(\tau)$ and show that

\begin{lemma}
\label{lem:co}
The coefficients $A_{n}$ in the Taylor expansion of $\sum_{l \in \mathbb{Z}}\langle l| \Delta^F(\tau)|0\rangle$ (Eq.~\eq{expansion}) can be expressed as
\begin{align}\label{eq:A_n}
\begin{split}
A_n & = \frac{i^n}{n!}\sum_{\gamma=1}^\Gamma\sum_{n_1+\cdots +n_{2q-1}=n} \binom{n}{n_1 \cdots n_{2q-1}}\prod_{j'=1,\dots,2q-1}^{\rightarrow}(\tilde{\alpha}_{j'}D_{{j'}}(0))^{n_{j'}} H_\gamma(0)\\
& -\frac{i^n}{n!}\sum_{j=1}^q(\beta_j+\alpha_j)\sum_{n_1+\cdots n_{2j-1}=n} \binom{n}{n_1 \cdots n_{2j-1}}\prod_{j'=1,\dots,2j-1}^{\rightarrow}(\tilde{\alpha}_{j'}D_{{j'}}(0))^{n_{j'}} H_{\gamma_j}(0)\\
& + \frac{i^n}{n!}\sum_{j=1}^q\beta_j\sum_{n_1+\cdots n_{2j-2}=n} \binom{n}{n_1 \cdots n_{2j-2}}\prod_{j'=1,\dots,2j-2}^{\rightarrow}(\tilde{\alpha}_{j'}D_{{j'}}(0))^{n_{j'}} H_{\gamma_j}(0).
\end{split}
\end{align}
\end{lemma}

\begin{lemma}\label{lem:remainder}
The remainder $A_{p_0}(\tau)$ in the Taylor expansion of $\sum_{l \in \mathbb{Z}}\langle l| \Delta^F(\tau)|0\rangle$ (Eq.~\eq{expansion}) satisfies
\begin{align}\label{eq:remainder2}
\int_0^\delta\dd{\tau}\norm{A_{p_0}(\tau)}\leq  3 (2q(p_0+1)kg\delta)^{p_0+1}(2k)^{-1}N.
\end{align}
\end{lemma}

The equality sign rather than inequality in \eq{A_n} is crucial since we need to append a projector to both sides. 
Note that in the first term of $A_n$ the summation over $\sum_{\gamma=1}^\Gamma H_{\gamma}(0)$ is one of the $\frac{q}{\Gamma}$ stages in $\sum_{j=1}^qH_{\gamma_j}(0)$, in the second and last terms the summation over $n_{j'\leq2j-1}$ is also contained in $n_{j'\leq2q-1}$, we have 
\begin{align*}
\norm{A_n\Pi_{\sigma}(0)} & \leq \frac{3}{n!}\sum_{j=1}^q\sum_{n_1+\cdots +n_{2q-1}=n} \binom{n}{n_1 \cdots n_{2q-1}}\norm{\prod_{j'=1,\dots,2q-1}^{\rightarrow}(\tilde{\alpha}_{j'}D_{{j'}}(0))^{n_{j'}} H_{\gamma_j}(0)\Pi_\sigma(0)}\\
& \leq 3\sum_{j=1}^q\sum_{n_1+\cdots +n_{2q-1}=n}\norm{\prod_{j'=1,\dots,2q-1}^{\rightarrow}(\tilde{\alpha}_{j'}D_{{j'}}(0))^{n_{j'}} H_{\gamma_j}(0)\Pi_\sigma(0)}\\
& \leq 3\sum_{j=1}^q\sum_{j_1',\dots,j_n'=1}^{2q-1}\norm{D_{j_1'}(0)\cdots D_{j_n'}(0)H_{\gamma_j}(0)\Pi_\sigma(0)}\\
&\leq 3\left(\frac{q} {\Gamma}\right)^{n+1}\sum_{\gamma_1,\dots,\gamma_n=1}^{\Gamma+1}\sum_{\gamma_{0}=1}^\Gamma\norm{\prod_{j=1}^n\mathcal{D}_{\gamma_j}(0)H_{\gamma_{0}}(0)\Pi_{\sigma}(0)},
\end{align*}
where the second inequality follows that the multinomial coefficient is no larger than $n!$, the third and last inequality follows the repetition relationship and $|\tilde{\alpha}_{j'}|\leq1$. 
We can further split the norm into $(n+1)$-layers with local interaction terms as in \sec{scaling}:
\begin{align*}
\norm{A_n\Pi_{\sigma}(0)} & \leq 3\left(\frac{q} {\Gamma}\right)^{n+1}\sum_{j=1}^{n+1}(2\Gamma)^{n+1-j}\sum_{X_1,\dots,X_j}\frac{j^{n+1}}{j!}F^j\norm{[h_{X_j},\dots,[h_{X_2},h_{X_1}]]\Pi_\sigma(0)}\\
& \leq  3 q^{n+1} \sum_{j=1}^{n+1}2^{n+1-j}\frac{j^{n+1}}{j!}F^j\sum_{X_1,\dots,X_j}\big\|{\underbrace{[h_{X_j},\dots,[h_{X_2},h_{X_1}]]}_{\text{commutator}}\Pi_\sigma(0)}\big\|,
\end{align*}

Note that in the above analysis, all the operators are considered to be at the same time $t=0$.
As we mentioned in \lem{LeakOp}, at an instantaneous time, the spectral projector is equivalent to $\Pi_{\sigma}(0)=\Pi_{\leq\Delta_{t=0}}$ for some energy threshold $\Delta_{t=0}$. 
Similarly we define $\Pi_{>\Delta_{t=0}'}$ ($\Delta_{t=0}'\geq\Delta_{t=0}$) for the fixed spectrum of $H(0)$. We then use the techniques for the time-independent case~\cite{Mizuta_2025}:
\begin{align*}
\norm{\text{commutator}\cdot \Pi_{\leq\Delta_{t=0}}}\leq \norm{\Pi_{>\Delta_{t=0}'}\cdot\text{commutator}\cdot \Pi_{\leq\Delta_{t=0}}}+\underbrace{\norm{\Pi_{\leq\Delta_{t=0}'}\cdot\text{commutator}\cdot\Pi_{\leq\Delta_{t=0}}}}_{\leq\norm{\Pi_{\leq\Delta_{t=0}'}\cdot\text{commutator}\cdot\Pi_{\leq\Delta_{t=0}'}}}.
\end{align*}

The sum of terms of the first kind is directly bounded by~\lem{LeakOp}:

\begin{align*}
\norm{\Pi_{>\Delta_{t=0}'}A_n\Pi_{\leq\Delta_{t=0}}}\leq & 3q^{n+1}\sum_{j=1}^{n+1}2^{n+1-j}\frac{j^{n+1}}{j!}F^j\sum_{X_1,\dots,X_j}\norm{\Pi_{>\Delta_{t=0}'}[h_{X_j},\dots,[h_{X_2},h_{X_1}]]\Pi_{\leq\Delta_{t=0}}}\\
\leq & 3q^{n+1}\sum_{j=1}^{n+1}(2j)^n(kg)^{j-1}\cdot gN \cdot e^{-\lambda(\Delta_{t=0}'-\Delta_{t=0}-2g(k+(j-1)(k-1))}\\
\leq & 3(2eq(n+1)kg)^{n+1}(2k)^{-1}N e^{-\lambda(\Delta_{t=0}'-\Delta_{t=0})},
\end{align*}
where the second inequality follows that the commutator $[h_{X_j},\dots,[h_{X_2},h_{X_1}]]$ is supported on at most $k+(j-1)(k-1)$ spins. 

For the sum of terms of the second kind, we alternatively assign the function coefficients into the nested commutator, there are in total $S(n+1,j)$ tuples of $(m_j,\dots,m_1)$ after split:
\begin{align*}
& \sum_{X_1,\dots,X_j}\sum_{n_1+\cdots +n_j=n+1-j}\Big|\frac{\dd^{n_j}}{\dd{t}^{n_j}}(f_{X_j}(0),\dots,\frac{\dd^{n_1}}{\dd{t}^{n_1}}f_{X_1}(0))\Big|\norm{\Pi_{\leq\Delta_{t=0}'}[h_{X_j},\dots,[h_{X_2},h_{X_1}]]\Pi_{\leq\Delta_{t=0}'}} \\ 
= & \sum_{S(n+1,j)}\sum_{X_1,\dots,X_j} \left\|\Pi_{\leq\Delta_{t=0}'}[f_{X_j}^{(m_j)}(0)h_{X_j},\dots,[f_{X_2}^{(m_2)}(0)h_{X_2},f_{X_1}^{(m_1)}(0)h_{X_1}]]\Pi_{\leq\Delta_{t=0}'}\right\|
\end{align*}

To use \lem{lowcommutator}, we first note that by definition the layer of the nested commutator consisting of no derivatives already satisfies the requirement 
\begin{align*}
\left\|\Pi_{\leq\Delta_{t=0}'}\sum\nolimits_X f_X(0)h_X \Pi_{\leq\Delta_{t=0}'}\right\| = \left\|\Pi_{\leq\Delta_{t=0}'}H(0) \Pi_{\leq\Delta_{t=0}'}\right\|\leq\Delta_{t=0}'.
\end{align*}
We now invoke \assum{2}, stated at the beginning of this section and motivated there by the adiabatic scaling $f_X^{(n)}(t)=T^{-n}\partial_s^n a_X(s)$. Then by \lem{lowcommutator} we have
\begin{align*}
\norm{\Pi_{\leq\Delta_{t=0}'}A_n\Pi_{\leq\Delta_{t=0}'}}\leq & 3q^{n+1}\sum_{j=1}^{n+1}2^{n+1-j}S(n+1,j)(j-1)!(2kg)^{j-1}\Delta_{t=0}'\\
\leq &3q^{n+1} \sum_{j=1}^{n+1} (2j)^n (kg)^{j-1}\Delta_{t=0}'\\
\leq & 3(n+1)^{n+1}(2qkg)^n q\Delta_{t=0}'.
\end{align*}
Combining the two parts together, note that for $n$ from $p$ to $p_0-1$, $n+1\leq p_0$,
\begin{align}
\begin{split}\label{eq:sum}
\sum_{n=p}^{p_0-1}\frac{\delta^{n+1}}{n+1}\norm{A_n\Pi_\sigma(0)} & \leq 3\sum_{n=p}^{p_0-1}(2eqp_0kg\delta)^{n+1}(2k)^{-1}Ne^{-\lambda(\Delta_{t=0}'-\Delta_{t=0})} + 3\sum_{n=p}^{p_0-1}(2qp_0kg\delta)^n q\Delta_{t=0}'\delta\\
& \leq 3\sum_{n=p}^{\infty}(2eqp_0kg\delta)^{n+1}(2k)^{-1}Ne^{-\lambda(\Delta_{t=0}'-\Delta_{t=0})} + 3\sum_{n=p}^{\infty}(2qp_0kg\delta)^n q\Delta_{t=0}'\delta.
\end{split}
\end{align}
For a sufficiently small time step $\delta$ such that $2q(p_0+1)kg\delta\leq \frac{1}{2e}$, which can be achieved by multiplying $r$ with a constant, the formula for the sum of geometric sequences gives:
\begin{align*}
\sum_{n=p}^\infty(2eqp_0kg\delta)^n &\leq \sum_{n=p}^{\infty}\left(\frac{1}{2}\right)^n < 1;\\
\sum_{n=p}^\infty(2qp_0kg\delta)^n &= \frac{(2qp_0kg\delta)^p}{1-2qp_0kg\delta}< 2 (2qp_0kg\delta)^p.
\end{align*}
Respectively applying these to the two terms in~\eq{sum} gives a simplified upper bound:
\begin{align*}
\sum_{n=p}^{p_0-1}\frac{\delta^{n+1}}{n+1}\norm{A_n\Pi_\sigma(0)}\leq \frac{3}{2}k^{-1}Ne^{-\lambda(\Delta_{t=0}'-\Delta_{t=0})} + 6(2qp_0kg\delta)^pq\Delta_{t=0}'\delta.
\end{align*}

Finally, the remainder $A_{p_0}(\tau)$ in~\eq{remainder2} satisfies
\begin{align*}
\int_0^\delta\dd{\tau}\norm{A_{p_0}(\tau)}\leq  \frac{3}{2} (2q(p_0+1)kg\delta)^{p_0+1}k^{-1}N\leq \frac{3}{2}\left(\frac{1}{e}\right)^{p_0+1}k^{-1}N.
\end{align*}
Now we set the parameters as follows, such that the first summand and the remainder are both less than $\frac{\epsilon}{4r}$:
\begin{align*}
\Delta_{t=0}'&=\Delta_{t=0}+2gk\cdot \log (6k^{-1}Nr/\epsilon), \\
p_0 &= \max\{\lceil \log(6k^{-1}Nr/\epsilon)\rceil-1,~p\}.
\end{align*}

To sum up, the short-time simulation error \eqref{eq:form} of the generalized product formula is bounded by
\begin{align}
\norm{(U_p(\delta, 0)-V(\delta, 0))\Pi_{\sigma}(0)} & \leq \frac{3}{2}k^{-1}Ne^{-\lambda(\Delta_{t=0}'-\Delta_{t=0})} + 6(2qp_0kg\delta)^pq\Delta_{t=0}'\delta + \frac{3}{2}\left(\frac{1}{e}\right)^{p_0+1}k^{-1}N\nonumber \\
& \leq \frac{\epsilon}{2r} + 3(kp_0g)^p\Delta_{t=0}'(2q\delta)^{p+1}. \label{eq:shorterror}
\end{align}
Since $\Delta_{t=0}'$ completely gets rid of $\operatorname{poly}(N)$ scaling, the only remaining issue is non-commuting in step accumulation. For each short time segment, we can always substitute the initial threshold $\Delta_{t=0}$ with the energy supremum $\Delta$ over $[0,T]$.

We remark that the error of the standard product formula only differs from above by the absence of a constant $2$ before $\Gamma$. 
Therefore, the error bound~\eq{shorterror} also applies to standard product formulas.

\subsection{Long-time simulation error}\label{sec:longtime}
Next, we address the long-time simulation problem \eqref{longtime}. 
We cannot take for granted that
\begin{align*}
\left\|(U_p(T,0)-V(T,0)) \Pi_{\sigma}(0)\right\|\overset{?}{\leq} r\left\|(U_p(\delta,0)-V(\delta,0)) \Pi_{\sigma}(0)\right\|.
\end{align*}
We follow the decomposition in proving the triangle inequality for telescoping operator products:
\begin{equation*}
\begin{aligned}
& \norm{(U_p(T,0)-V(T,0)) \Pi_{\sigma}(0)}=\norm{\sum_{j=0}^{r-1}U_p(T,(j+1)\delta)\left(U_p((j+1)\delta,j\delta)-V((j+1)\delta,j\delta)\right)V(j\delta,0) \Pi_{\sigma}(0)} \\
& \leq \sum_{j=0}^{r-1}\norm{U_p(T,(j+1)\delta)\left(U_p((j+1)\delta,j\delta)-V((j+1)\delta,j\delta)\right)\cdot \underbrace{I}_{\text{split}}\cdot V(j\delta,0) \Pi_{\sigma}(0)}\\
& \leq \sum_{j=0}^{r-1}\norm{U_p(T,(j+1)\delta)\underbrace{\left(U_p((j+1)\delta,j\delta)-V((j+1)\delta,j\delta)\right)\Pi_{\sigma}(j\delta)}_{\text{short-time error}}V(j\delta,0) \Pi_{\sigma}(0)}\\
& + \sum_{j=0}^{r-1}\norm{U_p(T,(j+1)\delta)\left(U_p((j+1)\delta,j\delta)-V((j+1)\delta,j\delta)\right)\underbrace{(I-\Pi_\sigma(j\delta))V(j\delta,0) \Pi_{\sigma}(0)}_{\text{leakage error}}}.
\end{aligned}
\end{equation*}
The above analysis shows that the long-time simulation error is the sum of short-time simulation errors plus additional leakage error terms, that is,
\begin{align*}
\norm{(U_p(T,0)-V(T,0)) \Pi_{\sigma}(0)} \leq & \sum_{j=0}^{r-1}\norm{U_p((j+1)\delta,j\delta)-V((j+1)\delta,j\delta)\Pi_{\sigma}(j\delta)}\\
+ & \sum_{j=0}^{r-1}\norm{U_p((j+1)\delta,j\delta)-V((j+1)\delta,j\delta)}\cdot \norm{(I-\Pi_\sigma(j\delta))V(j\delta,0) \Pi_{\sigma}(0)}.
\end{align*}
Each short-time simulation error can be bounded in a similar way to~\eqref{eq:shorterror}. 
We only need to substitute the initial time $0$ with $j\delta$. 
Correspondingly, the energy threshold $\Delta_{t=0}$ should be modified for the fixed spectrum of $H(j\delta)$ with $\Pi_\sigma(j\delta)$, which is then bounded by the supremum $\Delta$.

Next, we use adiabatic perturbation theory with multiple eigenstates to bound the leakage error. 
For each leakage error term, we define the scaled time $s_j=\frac{t}{j\delta}\in [0,1]$. 
Note that the norm of derivative with respect to $s_j$ is smaller than that with respect to $s=\tfrac{t}{T}$ since $\dv{}{s_j}= j\delta \cdot \dv{}{t}$, $j\delta\leq T$. 
Therefore, every $\dot{H}(s_j)$ and $\ddot{H}(s_j)$ can be respectively upper bounded by
\begin{align*}
\|{\dot{H}}\|\equiv \max_{s\in[0,1]} \norm{\dv{H}{s}},\quad \|{\ddot{H}}\| \equiv \max_{s\in[0,1]}  \norm{\dv[2]{H}{s}}.
\end{align*}
Additionally, with a uniform lower bound $\gamma$ of spectral gap $\gamma(s)$, by \lem{adiabatic} we have
\begin{align}\label{eq:gapbound}
\norm{(I-\Pi_{\sigma}(j\delta))V(j\delta,0)\Pi_{\sigma}(0)} \leq \frac{1}{j\delta}\left(2\sigma\frac{\|{\dot{H}}\|}{\gamma^2}+7\sigma\sqrt{\sigma} \frac{\|{\dot{H}}\|^2}{\gamma^3}+\sigma\frac{\|{\ddot{H}}\|}{\gamma^2}\right).
\end{align}

Finally, we conclude that the whole long-time simulation error is bounded by 
\begin{align}
& \norm{(U_p(T,0)-V(T,0)) \Pi_{\sigma}(0)}\nonumber\\
\leq & \tilde{O}\left(g^{p+1}N\delta^{p+1}\right)\cdot \frac{1}{\delta}\left(1+\frac{1}{2}+\cdots \frac{1}{r-1}\right)\left(2\sigma\frac{\|{\dot{H}}\|}{\gamma^2}+7\sigma\sqrt{\sigma} \frac{\|{\dot{H}}\|^2}{\gamma^3}+\sigma\frac{\|{\ddot{H}}\|}{\gamma^2}\right)+ \tilde{O}\left(g^{p}\Delta'\frac{T^{p+1}}{r^p}\right)+\frac{\epsilon}{2}\nonumber\\
\leq & \tilde{O}\left(g^{p+1}N\frac{T^{p}}{r^{p}}\right)\cdot(1+\ln r)\left(2\sigma\frac{\|{\dot{H}}\|}{\gamma^2}+7\sigma\sqrt{\sigma} \frac{\|{\dot{H}}\|^2}{\gamma^3}+\sigma\frac{\|{\ddot{H}}\|}{\gamma^2}\right)+ \tilde{O}\left(g^{p}\Delta'\frac{T^{p+1}}{r^p}\right)+\frac{\epsilon}{2}, \label{eq:longtime_err_general}
\end{align}
where the last inequality follows that $\sum_{j=1}^{r-1} \frac{1}{j}\leq \ln(r)+1.$

\section{Applications} \label{sec:application}
We now apply the general error expression~\eqref{eq:longtime_err_general} to spin models in the adiabatic regime. This is the regime in which the gap-dependent leakage term has a direct interpretation and in which \assum{2} can be verified from the slow variation of the interpolation path.
In our analysis we take $k=O(1)$ and $p=O(1)$. The upper bound of interaction strength on a single qubit is set to \eqref{eq:g}: 
\begin{align*}
g=F\|\hspace{-.1em}|h|\hspace{-.1em}\|_1=O(\|\hspace{-.1em}|h|\hspace{-.1em}\|_1).
\end{align*}

For geometrically-local Hamiltonians, $\Gamma=O(1)$ and thus $q=O(1)$. 
We have a constant degree
\begin{align*}
\quad g=O(1).
\end{align*}

For power-law Hamiltonians, the recursive decomposition scheme in~\cite{Low_2023} gives $\Gamma=O(\log N)$, and thus $q=O(\log N)$. 
Appendix F of~\cite{childs2021theory} directly calculates the induced 1-norm such that
\begin{align*}
g= \begin{cases}{O}\left(N^{1-\alpha / D}\right), & \text { for } 0 \leq \alpha<D, \\ {O}(\log N), & \text { for } \alpha=D, \\ {O}(1), & \text { for } \alpha>D.\end{cases}
\end{align*}

Next, we choose an appropriate Trotter number $r$ to ensure the whole long-time simulation error is no more than $\epsilon$. 
In the first case, if the sum of step errors is the dominant term, that is,
\begin{align*}
3(2q)^{p+1}(p_0g)^p\Delta'\frac{T^{p+1}}{r^p}\leq\frac{\epsilon}{2}.
\end{align*}
For both geometrically-local and power-law Hamiltonians, we have
\begin{align*}
r = \tilde{O}\left(g\frac{(\Delta+g\log(N/\epsilon))^{1/p}T^{1+1/p}}{\epsilon^{1/p}}\right).
\end{align*}
This expression can be obtained by replacing the full extensive energy scale $gN$ in the leading commutator contribution of the full Hilbert space product formula Trotter number with $\Delta+g\log(N/\epsilon)$. If the chosen low-energy subspace has $\Delta=O(1)$ and $g=O(1)$ for a geometrically local model, the explicit $N^{1/p}$ contribution in the leading term is replaced by a logarithmic dependence. If instead $\Delta$, $g$, or the adiabatic time $T$ scales polynomially with $N$, the final end-to-end scaling must include those dependencies.

In the second case, if the error containing the leakage term is dominant, that is,
\begin{align*}
(2k(p+1))^p\left(\sum_{j=1}^{p+1}\Gamma^{p+1-j}g^jN\frac{T^{p}}{r^{p}}\right)(1+\ln r)\left(2\sigma\frac{\|{\dot{H}}\|}{\gamma^2}+7\sigma\sqrt{\sigma} \frac{\|{\dot{H}}\|^2}{\gamma^3}+\sigma\frac{\|{\ddot{H}}\|}{\gamma^2}\right) \leq \frac{\epsilon}{2}.
\end{align*}
It suffices to choose
\begin{align*}
r = \tilde{O}\Bigg(g^{1+1/p}\frac{N^{1/p}T}{\epsilon^{1/p}}{\left(\sigma\frac{\|{\dot{H}}\|}{\gamma^2}+\sigma\sqrt{\sigma} \frac{\|{\dot{H}}\|^2}{\gamma^3}+\sigma\frac{\|{\ddot{H}}\|}{\gamma^2}\right)^{1/p}}\Bigg).
\end{align*}

\subsection{Adiabatic state preparation}
In the adiabatic regime where
\begin{align}\label{eq:adiabatic}
& T = \Omega\left(\frac{gN}{\Delta}\cdot\left(\sigma\frac{\|{\dot{H}}\|}{\gamma^2}+\sigma\sqrt{\sigma} \frac{\|{\dot{H}}\|^2}{\gamma^3}+\sigma\frac{\|{\ddot{H}}\|}{\gamma^2}\right)\right), 
\end{align}    
the full-space sufficient condition for \assum{2} discussed in \sec{analysis} is satisfied, and the first term in the Trotter number following from \eqref{eq:longtime_err_general}, induced by the sum of short-time simulation errors, is dominant. We thus have the following simplified bound:
\begin{theorem}[Trotter number for simulating adiabatic state preparation]\label{thm:atrotternumber}
Choose a smooth adiabatic interpolation $H(t)$ of the form \eqref{eq:adiabatic_interpolation} for which \assum{1} and \assum{2} hold. Its instantaneous Hamiltonians are $N$-qubit, $k$-local multi-linear Hamiltonians as in \eqref{eq:multilinear}, with the local interaction terms shifted to be positive semi-definite. Let $\Pi_\sigma(t)$ project onto the lowest $\sigma$ instantaneous eigenstates, assume that this subspace is separated from its complement by a uniform gap $\gamma>0$, and set $\Delta=\max_{t\in[0,T]}E_\sigma(t)$. Let the interaction strength on each qubit be upper bounded by $g$ as defined in \eqref{eq:g}.

For evolution times in the adiabatic regime satisfying
\begin{align*}
T = \Omega\left(\frac{gN}{\Delta}\cdot\left(\sigma\frac{\|{\dot{H}}\|}{\gamma^2}+\sigma\sqrt{\sigma}\frac{\|{\dot{H}}\|^2}{\gamma^3}+\sigma\frac{\|{\ddot{H}}\|}{\gamma^2}\right)\right),
\end{align*}
it suffices to choose
\begin{equation*}
r = \tilde{O}\left(g\frac{(\Delta+g\log(N/\epsilon))^{1/p}T^{1+1/p}}{\epsilon^{1/p}}\right),
\end{equation*}
so that both the $p$-th order generalized and standard time-dependent product formulas satisfy
\begin{align*}
\left\|(U_p(T,0)-V(T,0))\Pi_\sigma(0)\right\|&\leq\epsilon, &
\left\|(S_p(T,0)-V(T,0))\Pi_\sigma(0)\right\|&\leq\epsilon,
\end{align*}
where $V(T,0)=\mathcal{T}\exp(-i\int_0^T H(\tau)\dd{\tau})$. Thus, every initial state supported on $\Pi_\sigma(0)$ is simulated within error $\epsilon$. The leading product-formula contribution replaces the full extensive scale by the effective low-energy scale $\Delta+g\log(N/\epsilon)$.
\end{theorem}

\section{Numerical Experiments}
\subsection{Experimental setup}

We supplement the theoretical bounds with a numerical simulation for adiabatic state preparation. 
The calculations were performed on a MacBook Air equipped with an eight-core Apple M2 processor and 16 GB of memory, using Python 3.11.12.
Specifically, we consider a time-dependent $X$-initial nearest-neighbor Heisenberg interpolation
\begin{align*}
H(s)=-\sum_{i=1}^{N}X_i+sJ\sum_{i=1}^{N-1}(X_iX_{i+1}+Y_iY_{i+1}+Z_iZ_{i+1})+s\eta\sum_{i=1}^{N}Z_i,\quad s=t/T,
\end{align*}
with initial state $\ket{\psi_0(0)}=\ket{+}^{\otimes N}$. We set $J=0.25$ and $\eta=0.025$. The transverse field fixes a simple product ground state at the beginning of the path, the exchange interaction with strength $J$ introduces correlations, and the weak longitudinal field characterized by $\eta$ removes the special common-eigenstate structure of the untilted model. We choose $T=80$ by reading \eqref{eq:adiabatic} without the asymptotic $\Omega(\cdot)$ notation, for a spectral gap $\gamma$ on the order of $0.1$, to study finite-time adiabatic evolution. 
These parameters are kept fixed across all system sizes, product-formula orders, and diagnostics; only the displayed sweep variable is changed. Panels (a), (b), and (d) use $N=4,\ldots,10$, while panel (c) fixes $N=6$ and sweeps the product-formula step number. We compare second- and fourth-order product formulas.

We split the product formula into three implementable blocks, generated respectively by the $X+XX$, $YY$, and $ZZ+Z$ terms; the last block is diagonal. The exact evolution $V$ is computed with the eighth-order DOP853 integrator using relative and absolute tolerances $10^{-12}$ and $10^{-14}$, respectively, and the product-formula evolution is denoted by $S$. We compare the full operator error $\norm{V-S}$ with the initial-state projected error $\norm{(V-S)\ket{\psi_0(0)}}$, and use phase-aligned state-vector distances to assess preparation of the final ground-state ray. The most computationally demanding step is the explicit construction of the full exact-evolution operator by the ODE solver.

\subsection{Numerical results}
Panel (a) of \fig{heis_numerics} shows that, for $N\geq4$, the projected error is much smaller than the full operator-norm error. For example, at $N=10$, $r=256$, and second order, $\norm{V-S}\approx1.88\times10^{-1}$ while $\norm{(V-S)\ket{\psi_0(0)}}\approx7.37\times10^{-3}$. Panel (b) converts this error separation into a reduction in the required Trotter number. For $\epsilon=0.01$ and $N=10$, the fourth-order formula gives $r_{\min}^{\mathrm{full}}=109$ and $r_{\min}^{\mathrm{proj}}=77$, while the second-order formula gives $r_{\min}^{\mathrm{full}}=1108$ and $r_{\min}^{\mathrm{proj}}=220$.

Panel (c) shows that, for $N=6$, the state-preparation error converges to $3.81\times10^{-4}$ as $r$ increases for both second- and fourth-order product formulas. Panel (d) identifies this limiting value with the intrinsic exact adiabatic error and shows that the corresponding error floor ranges from about $3.1\times10^{-4}$ at $N=4$ to $4.9\times10^{-4}$ at $N=10$.

Taken together, these observations are qualitatively consistent with \thm{atrotternumber}. Panels (a) and (b) exhibit, at finite size, the central distinction underlying the theorem: for an input supported on the initial low-energy sector, the projected simulation error relevant to its evolution can be substantially smaller than the worst-case full operator-norm error, so the same target accuracy can be reached with fewer product-formula steps. 
Moreover, \thm{atrotternumber} controls the product-formula error relative to the exact finite-time evolution $V$, not the intrinsic error of $V$ as an adiabatic state-preparation procedure. The agreement between the plateau in panel (c) and the exact-evolution error in panel (d) therefore separates these two sources of error: increasing $r$ suppresses the product-formula contribution, while the finite-$T$ adiabatic floor remains.

\begin{figure}[H]
    \centering
    \begin{subfigure}[t]{0.48\linewidth}
        \includegraphics[width=\linewidth]{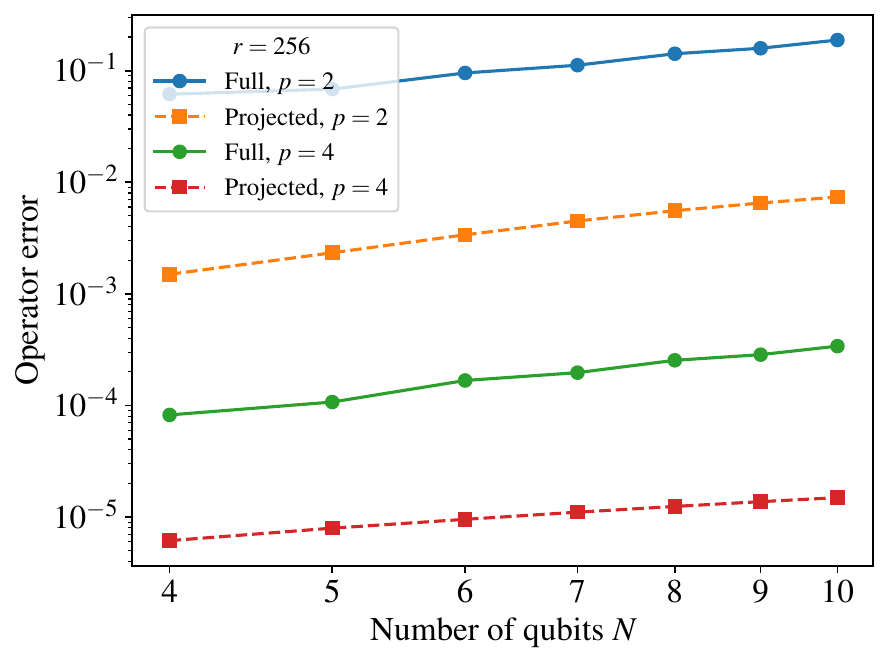}
        \caption{Operator norm errors at fixed $r=256$.}
        \label{fig:heis_operator}
    \end{subfigure}\hfill
    \begin{subfigure}[t]{0.48\linewidth}
        \includegraphics[width=\linewidth]{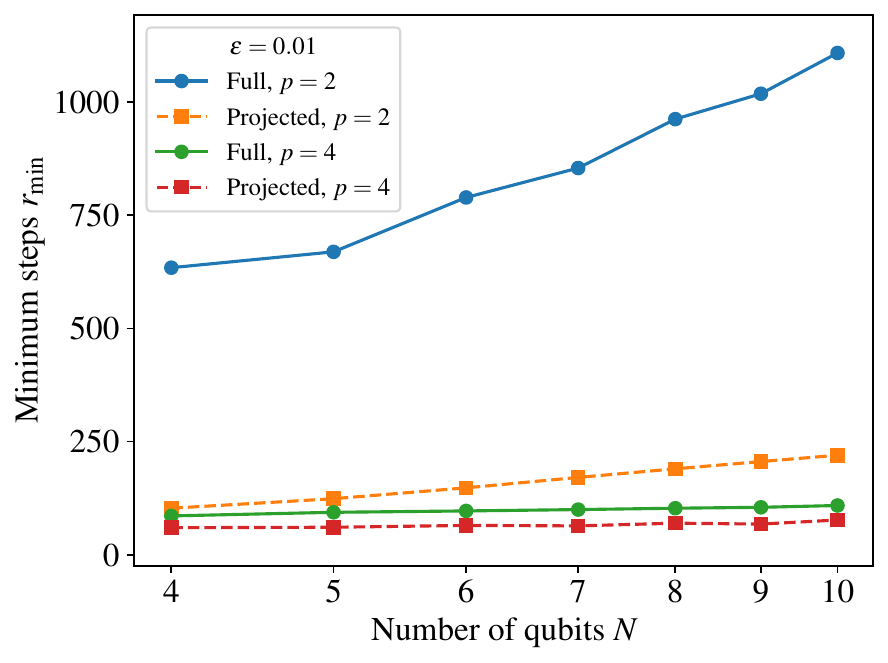}
        \caption{$r_{\min}$ required for target precision $\epsilon=0.01$.}
        \label{fig:heis_rmin}
    \end{subfigure}\par\medskip
    \begin{subfigure}[t]{0.48\linewidth}
        \includegraphics[width=\linewidth]{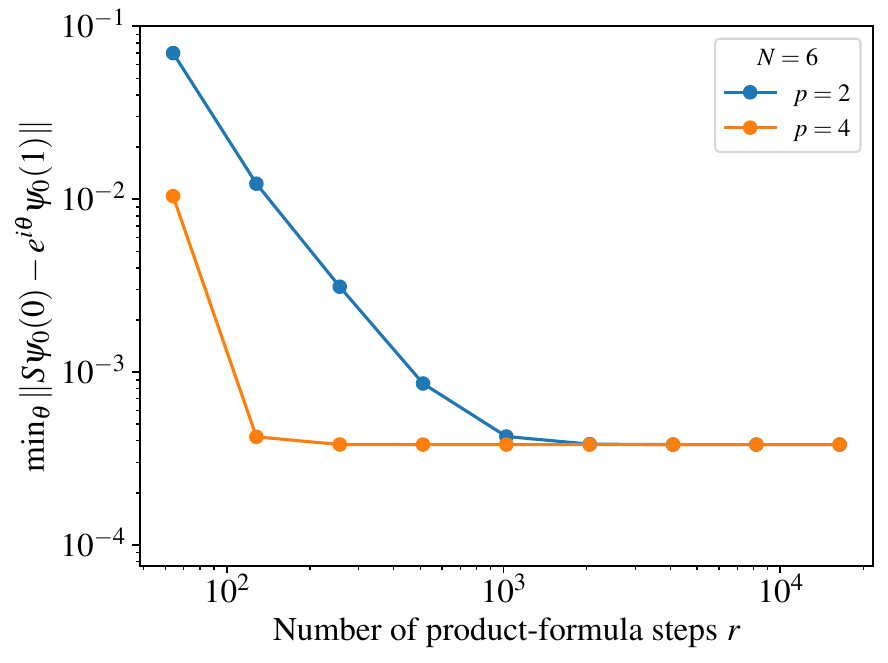}
        \caption{State-preparation error versus $r$ at $N=6$.}
        \label{fig:heis_state}
    \end{subfigure}\hfill
    \begin{subfigure}[t]{0.48\linewidth}
        \includegraphics[width=\linewidth]{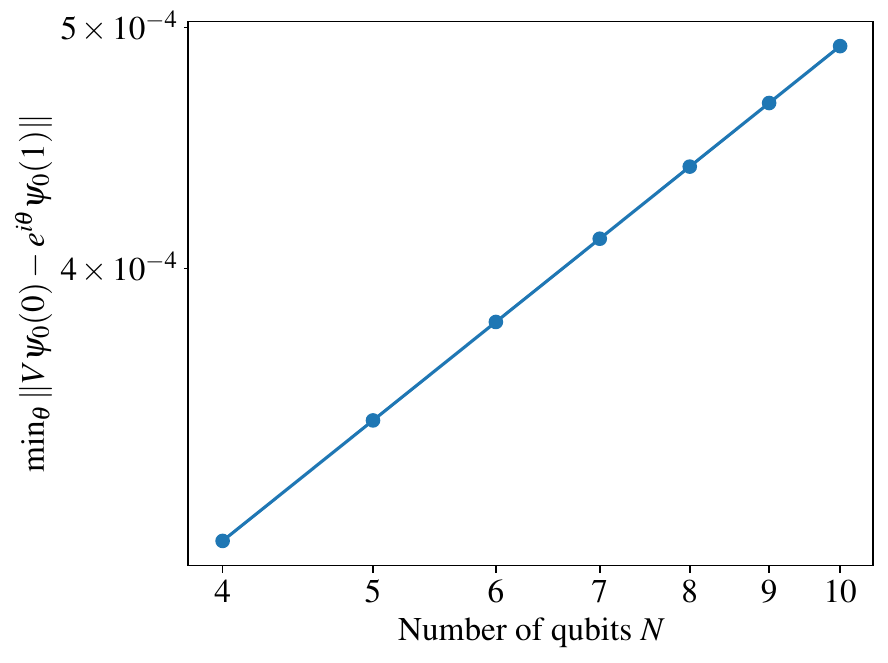}
        \caption{Intrinsic exact adiabatic error versus $N$.}
        \label{fig:heis_exact}
    \end{subfigure}
    \caption{Numerical experiments for the $X$-initial Heisenberg interpolation with a weak diagonal tilt. Panel (a) compares the full error $\norm{V-S}$ and projected error $\norm{(V-S)\ket{\psi_0(0)}}$ for second- and fourth-order product formulas. Panel (b) shows that the projected criterion reaches $\epsilon=0.01$ with fewer steps than the full operator criterion for both orders. Panels (c) and (d) report phase-aligned state-vector distances. This phase alignment removes the physically irrelevant dynamical/geometric phase accumulated during the adiabatic evolution, as well as the arbitrary phase convention of the numerically computed final eigenvector $\ket{\psi_0(1)}$. Panel (c) measures the state-preparation error $\min_\theta\norm{S\ket{\psi_0(0)}-e^{i\theta}\ket{\psi_0(1)}}$ relative to the target ray. Panel (d) depicts the intrinsic exact adiabatic error $\min_\theta\norm{V\ket{\psi_0(0)}-e^{i\theta}\ket{\psi_0(1)}}$, which sets the large-$r$ floor approached in panel (c).}
    \label{fig:heis_numerics}
\end{figure}

\FloatBarrier

\section{Lower Bound for Simulating Time-Dependent Hamiltonian} \label{sec:lower}
We remark that the lower bound of query complexity concerning generic time-dependent sparse Hamiltonian simulations is implicitly mentioned in the discussion section of~\cite{Berry_2020}. 
Here, we provide a rigorous proof for the full Hilbert space.
The hard instance combines the sparse-Hamiltonian parity construction of~\cite{Berry_2015} with the precision-dependent lower-bound argument of~\cite{berry2014exponential}, yielding a time-dependent construction that captures both the integrated max-norm contribution and the $\log(1/\epsilon)/\log\log(1/\epsilon)$ term.

\newpage
\begin{theorem}[Hard instance] \label{thm:lower}
For any $\epsilon, T>0$, integer $d \geq 2$, and fixed function of $\|H(t)\|_{\max }$ over $t\in[0,T]$, there exists a d-sparse time-dependent Hamiltonian $H(t)$ such that simulating $H(t)$ for time $t\in[0,T]$ within precision $\epsilon$ requires query complexity
\begin{align*}
\Omega\left(d\int_0^T\norm{H(\tau)}_{\max}\dd{\tau} +\frac{\log(1/\epsilon)}{\log\log(1/\epsilon)}\right).
\end{align*}
\end{theorem}
\begin{proof}
We consider the following hard instance
\begin{align*}
{H}(t) = \|H(t)\|_{\max }\cdot \tilde{H}/\|\tilde{H}\|_{\max},\quad 0\leq t\leq T,
\end{align*}
where $\tilde{H}$ follows the construction in Lemma 12 of~\cite{Berry_2015}. Given a string $x\in\{0,1\}^N$, let $d'=\lfloor \frac{d}{2}\rfloor$, we consider $2d'$-sparse Hamiltonian $\tilde{H}$ acts on $\ket{i,j,\ell}$ with $i \in\{0, \dots, N\}, j \in\{0,1\}$, and $\ell \in[d']$, whose non-zero entries are
\begin{equation*}
\langle i-1, j, \ell| \tilde{H}\left|i, j \oplus x_i, \ell^{\prime}\right\rangle=\left\langle i, j \oplus x_i, \ell^{\prime}\right| \tilde{H}|i-1, j, \ell\rangle=\sqrt{i(N-i+1)} / N.
\end{equation*}

Note that $[H(t), H(t')]=0$ holds for any time over $[0,T]$ due to linear scaling, the Dyson series or time-ordered exponential reduces to a case similar to the time-independent scenario, that is,
\begin{align*}
\mathcal{T}\exp\left(-i\int_0^T {H}(\tau)\dd{\tau}\right) = \exp\left(-i\int_0^T \norm{H(\tau)}_{\max} \dd{\tau}\cdot \tilde{H}/\|\tilde{H}\|_{\max}\right).
\end{align*}

We start the simulation from the state $|0,0, *\rangle$, where $|i, j, *\rangle:=\frac{1}{\sqrt{d'}} \sum_{\ell}|i, j, \ell\rangle$ denotes a uniform superposition over the third register. 
The subspace $\operatorname{span}\{|i, j, *\rangle: i \in\{0, \dots, N\}$, $j \in\{0,1\}\}$ is an invariant subspace of $\tilde{H}$. 
Since the initial state lies within this subspace, the quantum walk remains confined to this subspace. 
The non-zero matrix elements of $\tilde{H}$ in this invariant subspace are
\begin{gather*}
\langle i-1, j, *| \tilde{H}\left|i, j \oplus x_i, *\right\rangle=\left\langle i, j \oplus x_i, *\right| \tilde{H}|i-1, j, *\rangle=d' \sqrt{i(N-i+1)} / N.
\end{gather*}

 Note that $\|\tilde{H}\|_{\max}=\Theta(1)$, simulating ${H}(t)$ over time $[0,T]$ yields an unbounded-error algorithm determining the parity of string $p(x)=x_1\oplus x_2\oplus\cdots\oplus x_N$ since
\begin{align*}
\bra{N, p(x), *}\mathcal{T}\exp\left(-i\int_0^T {H}(\tau)\dd{\tau}\right)|0,0, *\rangle = \left(\sin (d'\int_0^T \norm{H(\tau)}_{\max}/(N\cdot \|\tilde{H}\|_{\max}))\right)^N,
\end{align*}
while 
\begin{align*}
\bra{N, p(x)\oplus 1, *}\mathcal{T}\exp\left(-i\int_0^T {H}(\tau)\dd{\tau}\right)|0,0, *\rangle = 0.
\end{align*}

If we measure $\mathcal{T}e^{-i\int_0^T {H}(\tau)\dd{\tau}}|0,0, *\rangle$ using the computational basis for the first two registers and an orthogonal basis containing $\ket{*}$ for the third register, then the third register must be exactly $\ket{*}$ due to aforementioned invariant subspace. 
If the first register is not $N$, we output $0$ or $1$ with equal probability. If the first register is $N$, we output the value of the second register. This is an unbounded-error algorithm of computing parity $p(x)$ with success probability $>\frac{1}{2}$, and thus requires $\Omega(N)$ queries~\cite{beals1998quantumlowerboundspolynomials,farhi1998limit}.

Since we allow that the implemented quantum circuit $U$ can deviate from $\mathcal{T}\exp(-i\int_0^T {H}(\tau)\dd{\tau})$ up to $\epsilon$ , it suffices to require~\cite{berry2014exponential}
\begin{align*}
\left(\sin (d'\int_0^T \norm{H(\tau)}_{\max}/(N\cdot \|\tilde{H}\|_{\max}))\right)^N\geq \epsilon.
\end{align*}

By taking the logarithm on both sides and using the approximation $\sin x\approx x$, it suffices to set
$$N=\Theta\left(\max\left\{d'\int_0^T\norm{H(\tau)}_{\max}\dd{\tau},\frac{\log(1/\epsilon)}{\log\log(1/\epsilon)}\right\}\right).$$

Above all, if we construct this kind of Hamiltonian with $N=\Theta(d\int_0^T\norm{H(\tau)}_{\max}\dd{\tau} +\frac{\log(1/\epsilon)}{\log\log(1/\epsilon)})$, simulating it over $[0,T]$ with error up to $\epsilon$ at least requires query complexity
\begin{align*}
\Omega(N)=\Omega\left(d\int_0^T\norm{H(\tau)}_{\max}\dd{\tau} +\frac{\log(1/\epsilon)}{\log\log(1/\epsilon)}\right).
\end{align*}

Thus, there is no generic quantum algorithm for time-dependent Hamiltonian simulations that can exceed this lower bound.
\end{proof}

We remark that the rescaled Dyson series~\cite{Berry_2020} achieves query complexity as a product of the above two parts, just as the former truncated Taylor series algorithm for the time-independent scenario. 
We hope that, in the future, a time-dependent analogue to quantum signal processing might match the lower bounds in all the above parameters.

\section*{Acknowledgements}
We thank Dong An, Boyang Chen, Yulong Dong, Minbo Gao, Lin Lin, Kaoru Mizuta, Burak Şahinoğlu, and Xinzhao Wang for helpful discussions. S.Z. and T.L. were supported by the National Natural Science Foundation of China (Grant Numbers 62372006 and 92365117).

\providecommand{\bysame}{\leavevmode\hbox to3em{\hrulefill}\thinspace}


\clearpage
\appendix

\section{Floquet Theory for Time-Dependent Hamiltonian Simulation}\label{append:floquet}
For a generic smooth time-dependent Hamiltonian $H(t)=\sum_{\gamma=1}^\Gamma H_\gamma(t)$ with each $H_\gamma(t)\in C^{p+2}$ over $t\in[0,T]$, Appendix A of~\cite{mizuta2024expliciterrorboundscommutator} extends each $H_\gamma(t)$ into a periodic Hamiltonian $H^{\ex}_\gamma(t)$ using a bump function as follows:
\begin{equation*}
H_\gamma^{\mathrm{ex}}(t)= \begin{cases}H_\gamma(t) & (t \in[0, T]) \\ \sum_{n=0}^{p+2} \frac{H_\gamma^{(n)}(t-T)^n}{n!} c\left(7-\frac{6 t}{T}\right) & \left(t \in\left(T, \frac{4}{3} T\right)\right) \\ 0 & \left(t \in\left[\frac{4}{3} T, \frac{5}{3} T\right]\right) \\ \sum_{n=0}^{p+2} \frac{H_\gamma^{(n)}(0)(t-2T)^n}{n!} c\left(\frac{6t}{T}-11\right) & \left(t \in\left(\frac{5}{3} T, 2 T\right)\right) ,\end{cases}
\end{equation*}
where the period $T^\ex=2T$ is independent of index $\gamma$, $H^{\ex}(t)=\sum_{\gamma=1}^\Gamma H^\ex_\gamma(t)$, and $c(t)$ is given by
\begin{equation*}
c(t)=\frac{\int_0^t \mathrm{~d} \tau b(\tau) b(1-\tau)}{\int_0^1 \mathrm{~d} \tau b(\tau) b(1-\tau)},\quad b(t)= \begin{cases}e^{-1 / t} & (t>0) \\ 0 & (t \leq 0)\end{cases}.
\end{equation*}

In this appendix, we abuse the notation $H_\gamma(t)$ to denote the periodic Hamiltonian without causing ambiguity, which admits a Fourier series:
\begin{align}
H_\gamma(t) =\sum_{m} H_{\gamma m}e^{-im\omega t},\quad \omega=\frac{2\pi}{T^\ex}=\frac{\pi}{T}.\label{eq:Fourier}
\end{align}
Given the $C^{p+2}$ condition, repeating integration by parts yields $\left\|H_{\gamma m}\right\| \in{O}\left(|m|^{-p-2}\right)$, which ensures absolute and uniform convergence of the $p$-th order derivatives series expansion. Similarly,
\begin{align*}
    H(t) =\sum_{m} H_{ m}e^{-im\omega t}, \quad H_m=\sum_{\gamma=1}^\Gamma H_{\gamma m}.
\end{align*}

Floquet theory relates the evolution operator of a periodic Hamiltonian $H(t)$ with a time-independent Floquet Hamiltonian defined on an infinite-dimensional space~\cite{levante1995formalized, Mizuta_2023}
\begin{align*}
H^F=\sum_{l \in \mathbb{Z}}\Big(\sum_m|l+m\rangle\langle l| \otimes H_m-l \omega|l\rangle\langle l| \otimes I\Big),
\end{align*}
where the ancilla system $\ket{l}$ contains the Fourier index $l\in \mathbb{Z}$. It is often separated into two terms $H^F=H^{\mathrm{Add}}-H_{\mathrm{LP}}$ with
\begin{align*}
    H^{\text{Add}}=\sum_m \text{Add}_m\otimes H_m,\quad \text{Add}_m=\sum_{l\in\mathbb{Z}}\ket{l+m}\bra{l},
\end{align*}
and the linear potential term
\begin{align*}
    H_{\mathrm{LP}}=\sum_{l\in\mathbb{Z}}l \omega|l\rangle\langle l| \otimes I.
\end{align*}

The exact evolution operator of $H(t)$, not necessarily over a whole period, is equivalent to 
\begin{align*}
V(t,0)=\sum_{l\in\mathbb{Z}} e^{-il\omega t} \langle l| e^{-i H^Ft}|0\rangle,
\end{align*}

On the other hand, one can similarly construct a Floquet Hamiltonian $H_\gamma^F$ for each $H_\gamma(t)$ by
\begin{align*}
    H^{F}_\gamma=H^{\text{Add}}_\gamma-H_{\mathrm{LP}},\quad H^{\text{Add}}_\gamma=\sum_m \text{Add}_m\otimes H_{\gamma m}.
\end{align*}
Then the $H^{F}$ can be split in the following two ways:
\begin{align}
H^F & =\sum_{\gamma=1}^{\Gamma} H_\gamma^{\mathrm{Add}}-H^{\mathrm{LP}}\label{eq:24} \\
& =\sum_{\gamma=1}^{\Gamma} H_\gamma^F+(\Gamma-1) H^{\mathrm{LP}} \label{eq:25}.
\end{align}

With the coefficients $\{\alpha_j\}$ and $\{\beta_j\}$ of the time-dependent generalized product formula~\eqref{eq:GPF}, Ref.~\cite{mizuta2024expliciterrorboundscommutator} defines the following time-independent product formula with respect to \eqref{eq:25}:
\begin{align*}
    T^F(t)
    & = e^{-i H_{\mathrm{LP}}\left(\beta_q+\alpha_q-1\right) t}e^{-i H_{\gamma_q}^F \alpha_q t} \prod_{j \leq q-1}^{\leftarrow}\left(e^{-i H_{\mathrm{LP}}\left(\beta_j+\alpha_j-\beta_{j+1}\right) t} e^{-i H_{\gamma_j}^F \alpha_j t}\right)e^{-iH_{\mathrm{LP}}\beta_1t}\\
    & =e^{-i H_{\gamma_q}^F \alpha_q t} \prod_{j \leq q-1}^{\leftarrow}\left(e^{-i H_{\mathrm{LP}}\left(\beta_j+\alpha_j-\beta_{j+1}\right) t} e^{-i H_{\gamma_j}^F \alpha_j t}\right),
\end{align*}
where the last equality follows end time $\beta_q+\alpha_q=1$ and start time $\beta_1=0$ since we focus on the Trotter-Suzuki formula. There are in total $2q-1$ exponentials. One can use simplified notations
\begin{align}
T^F(t)=\prod_{j^{\prime}=1,2,\dots,2q-1}^{\leftarrow}e^{-i\tilde{H}^F_{{j^{\prime}}}\tilde{\alpha}_{{j^{\prime}}}t},\quad
(\tilde{H}_{j^{\prime}}^F, \tilde{\alpha}_{j^{\prime}})= \begin{cases}(H_{\gamma_j}^F, \alpha_j) & (j^{\prime}=2 j-1) \\ (-H_{\mathrm{LP}}, \beta_{j+1}-\beta_j-\alpha_j) & (j^{\prime}=2 j)\end{cases}.\label{eq:tilde}
\end{align}
This product formula in the Floquet-Hilbert space is related to the original Hilbert space through

\begin{lemma}[Theorem 4 of~\cite{mizuta2024expliciterrorboundscommutator}] The generalized product formula~\eqref{eq:GPF} is equivalent to
    \begin{align*}
        U_p(t,0)=\sum_{l\in\mathbb{Z}} e^{-il\omega t} \langle l|T^F(t)|0\rangle,
    \end{align*}
\end{lemma}

As a corollary, by the variation-of-parameters formula, we can express the Trotter error by
\begin{align*}
U_p(t,0)-V(t,0) & =\sum_{l\in\mathbb{Z}} e^{-il\omega t} \langle l|( T^F(t)-e^{-i H^Ft})|0\rangle\\
& = i \sum_{l\in\mathbb{Z}} e^{-il\omega t} \langle l|\int_0^t\dd{\tau}  e^{-iH^F(t-\tau)} U_p^{^F}(\tau) \Delta^F(\tau)|0\rangle,
\end{align*}
where the residual term defined for the Floquet Hamiltonian is
\begin{align}
\Delta^F(\tau)= T^{F}(\tau)^\dagger H^FT^F(\tau) - iT^F(\tau)^\dagger\frac{\dd}{\dd \tau}T^F(\tau) .\label{eq:residual}
\end{align}
Specifically, for $T^F(\tau)$ in the form of \eqref{eq:tilde}, the residual term can be further calculated as
\begin{align*}
\hspace{-1em}\Delta^F(\tau) & =\sum_{\gamma=1}^{\Gamma} T^F(\tau)^{\dagger} H_\gamma^{\mathrm{Add}} T^F(\tau)\\ 
& -\sum_{j=1}^q\left(\beta_j+\alpha_j\right) T_{\leq 2j-1}^F(\tau)^{\dagger} H_{\gamma_j}^{\mathrm{Add}} T_{\leq 2j-1}^F(\tau)\\
& +\sum_{j=1}^q \beta_j T_{<2 j-1}^F(\tau)^{\dagger} H_{\gamma_j}^{\mathrm{Add}} T_{<2 j-1}^F(\tau).
\end{align*}

Moreover, by the translation symmetry, the error representation can be simplified to the form of~\eqref{eq:representation} in the main text:
\begin{lemma}[Variant of~{\cite[Theorem 8]{mizuta2024expliciterrorboundscommutator}}]\label{lem:representation}
    The additive error of the time-dependent product formula can be expressed as
\begin{align*}
U_p(t,0)-V(t,0) = i \int_0^t \mathrm{~d} \tau V(t, \tau) U_p(\tau, 0)\left(\sum_{l \in \mathbb{Z}}\langle l| \Delta^F(\tau)|0\rangle\right),
\end{align*}
\end{lemma}

Next, we calculate the coefficients of Taylor expansion of $\Delta^F(\tau)$ to prove \lem{co}, and bound the remainder to prove \lem{remainder} in the main text.

\subsection{Proof of \lem{co}}
\begin{proof}
Now we consider the $(p_0-1)$-th order Taylor expansion of $\sum_{l \in \mathbb{Z}}\langle l| \Delta^F(\tau)|0\rangle$ \eqref{eq:expansion} with the integral form of the remainder. The order condition ensures that $\tau^0,\dots,\tau^{p-1}$ terms cancel, then
\begin{gather*}
    \Delta^F(\tau) =\sum_{n=p}^{p_0-1}\frac{\Delta^{F}(0)^{(n)}}{n!}\tau^n+\int_0^\tau\dd{\tau}_1 \frac{(\tau-\tau_1)^{p_0-1}}{(p_0-1)!}\Delta^{F}(\tau_1)^{(p_0)},\\
     \sum_{l \in \mathbb{Z}}\langle l| \Delta^F(\tau)|0\rangle=\sum_{n=p}^{p_0-1}\frac{1}{n!}\sum_{l\in \mathbb{Z}}\langle l|\Delta^F(0)^{(n)}|0\rangle\cdot \tau^n +\int_0^\tau\dd{\tau}_1 \frac{(\tau-\tau_1)^{p_0-1}}{(p_0-1)!}\sum_{l\in \mathbb{Z}}\langle l|\Delta^{F}(\tau_1)^{(p_0)}|0\rangle.
\end{gather*}
We remark that $e^{\tau A}Be^{-\tau A}$ can be rewritten as $e^{\tau\operatorname{ad}_A}(B)$ with the same Taylor series.
Then we can substantially simplify our calculation using the following notation 
\begin{align*}
T^F(\tau)^\dagger(\cdot)T^F(\tau)=\exp({i\tilde{\alpha}_{1}\tau\operatorname{ad}_{\tilde{H}^F_{{1}}}})\cdots\exp({i\tilde{\alpha}_{{2q-1}}\tau\operatorname{ad}_{\tilde{H}^F_{{2q-1}}}})(\cdot)=\prod^{\rightarrow}_{j'=1,\dots,2q-1}\exp({i\tilde{\alpha}_{{j^{\prime}}}\tau\operatorname{ad}_{\tilde{H}^F_{{j^{\prime}}}}})(\cdot).
\end{align*}
If the operator $(\cdot)$ is independent of time, the $n$-th order derivative of the above equation is
\begin{align*}
\hspace{-2em}\sum_{n_1+\cdots +n_{2q-1}=n} \binom{n}{n_1 \cdots n_{2q-1}} (i\tilde{\alpha}_{1}\operatorname{ad}_{\tilde{H}^F_{{1}}})^{n_1}\exp({i\tilde{\alpha}_{1}\tau\operatorname{ad}_{\tilde{H}^F_{{1}}}})\cdots(i\tilde{\alpha}_{{2q-1}}\operatorname{ad}_{\tilde{H}^F_{{2q-1}}})^{n_{2q-1}}\exp({i\tilde{\alpha}_{{2q-1}}\tau\operatorname{ad}_{\tilde{H}^F_{{2q-1}}}})(\cdot).
\end{align*}

Note that the derivative at $\tau=0$ only contains commutators since the exponential terms reduce to identity, we have:
\begin{align*}
\Delta^F(0)^{(n)}& = i^n\sum_{\gamma=1}^\Gamma\sum_{n_1+\cdots +n_{2q-1}=n} \binom{n}{n_1 \cdots n_{2q-1}}\prod_{j'=1,\dots,2q-1}^{\rightarrow}(\tilde{\alpha}_{j'}\operatorname{ad}_{\tilde{H}^F_{j'}})^{n_{j'}}(H_\gamma^{\mathrm{Add}})\\
& -i^n\sum_{j=1}^q(\beta_j+\alpha_j)\sum_{n_1+\cdots n_{2j-1}=n} \binom{n}{n_1 \cdots n_{2j-1}}\prod_{j'=1,\dots,2j-1}^{\rightarrow}(\tilde{\alpha}_{j'}\operatorname{ad}_{\tilde{H}^F_{j'}})^{n_{j'}}(H_{\gamma_j}^{\mathrm{Add}})\\
& + i^n\sum_{j=1}^q\beta_j\sum_{n_1+\cdots n_{2j-2}=n} \binom{n}{n_1 \cdots n_{2j-2}}\prod_{j'=1,\dots,2j-2}^{\rightarrow}(\tilde{\alpha}_{j'}\operatorname{ad}_{\tilde{H}^F_{j'}})^{n_{j'}}(H_{\gamma_j}^{\mathrm{Add}}).
\end{align*}
The commutator in the Floquet-Hilbert space is related to the original Hilbert space by
\begin{align*}
{\left[H_{\mathrm{LP}}, H_\gamma^{\mathrm{Add}}\right] } = \sum_m \operatorname{Add}_m \otimes\left(-i \frac{\mathrm{~d}}{\mathrm{~d} t} H_\gamma(t)\right)_m, \quad 
{\left[H_{\gamma^{\prime}}^{\mathrm{Add}}, H_\gamma^{\mathrm{Add}}\right] } =\sum_m \operatorname{Add}_m \otimes\left(\left[H_{\gamma^{\prime}}(t), H_\gamma(t)\right]\right)_m,
\end{align*}
where the subscript $m$ denotes the time-independent component in the Fourier expansion as $H_m$ in \eqref{eq:Fourier}. According to the definition of $\tilde{H}^F_{j'}$ in \eqref{eq:tilde}, 
\begin{align*}
\tilde{\alpha}_{j^{\prime}} \operatorname{ad}_{\tilde{H}_{j^{\prime}}^F} H_\gamma^{\mathrm{Add}}=\sum_m \operatorname{Add}_m \otimes\left(\tilde{\alpha}_{j^{\prime}} D_{j^{\prime}}(t) H_\gamma(t)\right)_m,\quad  
D_{j'}(t)=\begin{cases}\operatorname{ad}_{H_{\gamma_j}(t)}+i\frac{\dd}{\dd{t}} & (j^{\prime}=2 j-1) \\ i\frac{\dd}{\dd{t}} & (j^{\prime}=2 j)\end{cases}.
\end{align*}
Therefore, we can make the following substitution for terms in $\Delta^F(0)^{(n)}$:
\begin{align*}
\prod_{j'=1,\dots,2q-1}^{\rightarrow}(\tilde{\alpha}_{j'}\operatorname{ad}_{\tilde{H}^F_{j'}})^{n_{j'}}(H_\gamma^{\mathrm{Add}})\longrightarrow \sum_m\mathrm{Add_m\otimes }\left(\prod_{j'=1,\dots,2q-1}^{\rightarrow}(\tilde{\alpha}_{j'}D_{{j'}}(t))^{n_{j'}} H_\gamma(t)\right)_m.
\end{align*}
To calculate $\sum_{l\in \mathbb{Z}}\langle l|\Delta^F(0)^{(n)}|0\rangle$, we use the following equality
\begin{align*}
\sum_{l\in\mathbb{Z}}\sum_m\langle l| \mathrm{Add}_m |0\rangle \otimes H_m = \sum_{l,l_1,m} \langle l|l_1+m\rangle\langle l_1|0\rangle \otimes H_m= \sum_m H_m=H(0).
\end{align*}

Finally, the coefficients in \eqref{eq:expansion} are given by 
\begin{align*}
A_n &= \frac{1}{n!}\sum_{l\in \mathbb{Z}}\langle l|\Delta^F(0)^{(n)}|0\rangle \\
& = \frac{i^n}{n!}\sum_{\gamma=1}^\Gamma\sum_{n_1+\cdots +n_{2q-1}=n} \binom{n}{n_1 \cdots n_{2q-1}}\prod_{j'=1,\dots,2q-1}^{\rightarrow}(\tilde{\alpha}_{j'}D_{{j'}}(0))^{n_{j'}} H_\gamma(0)\\
& -\frac{i^n}{n!}\sum_{j=1}^q(\beta_j+\alpha_j)\sum_{n_1+\cdots n_{2j-1}=n} \binom{n}{n_1 \cdots n_{2j-1}}\prod_{j'=1,\dots,2j-1}^{\rightarrow}(\tilde{\alpha}_{j'}D_{{j'}}(0))^{n_{j'}} H_{\gamma_j}(0)\\
& + \frac{i^n}{n!}\sum_{j=1}^q\beta_j\sum_{n_1+\cdots n_{2j-2}=n} \binom{n}{n_1 \cdots n_{2j-2}}\prod_{j'=1,\dots,2j-2}^{\rightarrow}(\tilde{\alpha}_{j'}D_{{j'}}(0))^{n_{j'}} H_{\gamma_j}(0).
\end{align*}
\end{proof}

\subsection{Proof of \lem{remainder}}
\begin{proof}
For the remainder in the integral form, if we directly calculate the $p_0$-th order derivative at $\tau_1\neq0$, there are exponential terms interleaved between $\tilde{\alpha}_j \operatorname{ad}_{\tilde{H}^F_j}$, preventing us from expressing its norm by commutators. To address this issue, we alternatively adopt the iterative decomposition scheme in Theorem 10 of~\cite{childs2021theory}. 
By merely substituting the $p$ with $p_0$ in \eqref{eq:remainder}, we have
\begin{multline*}
    \hspace*{-1.5em}A_{p_0}(\tau)= \sum_{\gamma=1}^\Gamma\sum_{j'=1}^{2q-1}\sum_{\substack{n_1+\cdots+n_{j'}=p_0 \\ n_{j^{\prime}} \neq 0}}U_{< \lceil j'/2 \rceil}(\tau,0)^\dagger\int_0^\tau \dd{\tau}_1\frac{i^{p_0}\left(\tau-\tau_1\right)^{n_{j^{\prime}}-1} \tau^{p_0-n_{j^{\prime}}}}{\left(n_{j^{\prime}}-1\right)!n_{j^{\prime}-1}!\cdots n_{1}!} C_{j'\rightarrow 2q-1}(H_{\gamma})\cdot  U_{< \lceil j'/2 \rceil}(\tau,0)\\
- \sum_{j=1}^q(\beta_j+\alpha_j)\sum_{j'=1}^{2j-1}\sum_{\substack{n_1+\cdots+n_{j'}=p_0 \\ n_{j^{\prime}} \neq 0}}U_{< \lceil j'/2 \rceil}(\tau,0)^\dagger\int_0^\tau \dd{\tau}_1\frac{i^{p_0}\left(\tau-\tau_1\right)^{n_{j^{\prime}}-1} \tau^{p_0-n_{j^{\prime}}}}{\left(n_{j^{\prime}}-1\right)!n_{j^{\prime}-1}!\cdots n_{1}!} C_{j'\rightarrow 2j-1}(H_{\gamma_j})\cdot  U_{< \lceil j'/2 \rceil}(\tau,0)\\
+ \sum_{j=1}^q\beta_j\sum_{j'=1}^{2j-2}\sum_{\substack{n_1+\cdots+n_{j'}=p_0 \\ n_{j^{\prime}} \neq 0}}U_{< \lceil j'/2 \rceil}(\tau,0)^\dagger\int_0^\tau \dd{\tau}_1\frac{i^{p_0}\left(\tau-\tau_1\right)^{n_{j^{\prime}}-1} \tau^{p_0-n_{j^{\prime}}}}{\left(n_{j^{\prime}}-1\right)!n_{j^{\prime}-1}!\cdots n_{1}!} C_{j'\rightarrow 2j-2}(H_{\gamma_j})\cdot  U_{< \lceil j'/2 \rceil}(\tau,0),
\end{multline*}
Regarding its norm, since $C_{j'\rightarrow 2q-1}(A)$ depends on $\tau_1$, we cannot separate the operators and directly calculate $\int_{0}^\tau\dd{\tau_1}(\tau-\tau_1)^{n_j'-1}=\tfrac{1}{n_j'}\tau^{n_{j'}}$. Instead, we use the inequalities $(\tau-\tau_1)^{n_{j'}-1}\tau^{p_0-n_j'}\leq \tau^{p_0-1} $ and $\tfrac{1}{(n_{j'}-1)!}=\tfrac{n_{j'}}{n_{j'}!}\leq \tfrac{p_0}{n_{j'}!}$. Then we can take the maximum value in the final step:
\begin{align*}
\norm{A_{p_0}(\tau)}
& \leq 3\sum_{j=1}^q\sum_{j'=1}^{2q-1}\sum_{\substack{n_1+\cdots+n_{j'}=p_0 \\ n_{j^{\prime}} \neq 0}}\frac{p_0\tau^{p_0-1}}{n_{j'}!n_{j'-1}!\cdots n_1!}\int_0^\tau\dd{\tau_1}\norm{C_{j'\rightarrow 2q-1}(H_{\gamma_{j}})}\\
& \leq 3 p_0\tau^{p_0}\int_0^\tau\dd{\tau_1} \sum_{j=1}^q\sum_{j'=1}^{2q-1}\sum_{\substack{n_1+\cdots+n_{j'}=p_0 \\ n_{j^{\prime}} \neq 0}} \norm{C_{j'\rightarrow 2q-1}(H_{\gamma_{j}})}\\
& \leq 3 p_0\tau^{p_0} \max _{\tau_1 \in[0, \tau]}\left(\sum_{j=1}^{q} \sum_{j_1^{\prime}, \dots, j_{p_0}^{\prime}=1}^{2q-1}\left\|(|\tilde{\alpha}_{j_1'}|D_{j_1'}(\tau_1))\cdots(|\tilde{\alpha}_{j_{p_0}'}|D_{j_{p_0}'}(\tau_1))H_{\gamma_j}\left(\tau_1\right)\right\|\right) .
\end{align*}
We then reach a conclusion similar to the non-asymptotic Trotter error with a concrete prefactor:
\begin{align*}
\int_0^\delta\dd{\tau}\norm{A_{p_0}(\tau)}\leq 3\left(\frac{q}{\Gamma}\right)^{p_0+1} \cdot \max _{\tau \in[0, \delta]} \alpha_{\text {com }}^{p_0+1}(\tau) \cdot \delta^{p_0+1}.
\end{align*}
Further splitting the commutator into local iteration terms as in \sec{scaling}, we have
\begin{align*}
\max _{\tau_1 \in[0, \tau]}\alpha_{\text {com }}^{p_0+1}(\tau)\leq (2k(p_0+1))^{p_0}\cdot \left(\sum_{j=1}^{p_0+1}\Gamma^{p_0+1-j}g^{j}N\right).
\end{align*}
Finally, we upper bound the integral of the remainder as follows:
\begin{align*}
\int_0^\delta\dd{\tau}\norm{A_{p_0}(\tau)}\leq 3 (2q(p_0+1)kg\delta)^{p_0+1}(2k)^{-1}N.
\end{align*}
\end{proof}
\end{document}